\documentclass[11pt, onecolumn]{article}
\usepackage[top=1in, bottom=1in, left=1.25in, right=1.25in]{geometry}

\usepackage{amsfonts}
\usepackage{amsmath,amssymb}
\usepackage{graphicx}
\usepackage{color, soul}
\usepackage{algorithm,algorithmic}
\usepackage{bm}
\usepackage{booktabs}
\usepackage{flushend}
\usepackage{tikz}
\usetikzlibrary{arrows}
\usepackage{subfigure}

\usepackage[amsmath,thmmarks]{ntheorem}
\usepackage{theorem}

\newtheorem{cor}{Corollary}
\newtheorem{lem}{Lemma}
\newtheorem{defi}{Definition}
\newtheorem{rem}{Remark}
\newtheorem{thm}{Theorem}

\theoremheaderfont{\sc}\theorembodyfont{\upshape}
\theoremstyle{nonumberplain}
\theoremseparator{}
\theoremsymbol{\rule{1ex}{1ex}}
\newtheorem{proof}{Proof}

\newcommand{\aff}{{\rm aff}}
\newcommand{\oaff}{\overline{{\rm aff}}}
\newcommand{\set}{\mathcal}
\newcommand{\od}{\overline{{D}}}
\newcommand{\affX}{\aff_{\set X}}
\newcommand{\affXS}{\aff_{\set X}^2}
\newcommand{\affY}{\aff_{\set Y}}
\newcommand{\affYS}{\aff_{\set Y}^2}
\newcommand{\oaffYS}{\oaff_{\set Y}^2}

\newcommand{\dX}{D_{\set X}}
\newcommand{\dXS}{D_{\set X}^2}
\newcommand{\dY}{D_{\set Y}}
\newcommand{\dYS}{D_{\set Y}^2}
\newcommand{\odYS}{\od_{\set Y}^2}

\begin{document}

\title{Restricted Isometry Property of Gaussian Random Projection for Finite Set of Subspaces}

\author{Gen Li and Yuantao Gu%
\thanks{The authors are with the Department of Electronic Engineering and Tsinghua National Laboratory for Information Science and Technology (TNList), Tsinghua University, Beijing 100084, China. The corresponding author of this work is Y. Gu  (E-mail:\,gyt@tsinghua.edu.cn).}}

\date{submitted April 7, 2017, revised August 11, 2017, accepted November 8, 2017\\
This paper will be published on \emph{IEEE Transactions on Signal Processing.}}

\maketitle

\begin{abstract}
Dimension reduction plays an essential role when decreasing the complexity of solving large-scale problems.
The well-known Johnson-Lindenstrauss (JL) Lemma and Restricted Isometry Property (RIP) admit the use of random projection to reduce the dimension while keeping the Euclidean distance, which leads to the boom of Compressed Sensing and the field of sparsity related signal processing. 
Recently, successful applications of sparse models in computer vision and machine learning have increasingly hinted that the underlying structure of high dimensional data looks more like a union of subspaces (UoS). 
In this paper, motivated by JL Lemma and an emerging field of Compressed Subspace Clustering (CSC), we study for the first time the RIP of Gaussian random matrices for the compression of two subspaces based on the generalized projection $F$-norm distance. 
We theoretically prove that with high probability the \emph{affinity} or \emph{distance} between two projected subspaces are concentrated around their estimates.
When the ambient dimension after projection is sufficiently large, the affinity and distance between two subspaces almost remain unchanged after random projection.
Numerical experiments verify the theoretical work.

\end{abstract}

\section{Introduction}

In the big data era we confront with large-scale problems dealing with data points or features in high dimensional vector spaces. 
In the enduring effort of trying to decrease the complexity of solving such large problems, dimension reduction has played an essential role. 
The well-known Johnson-Lindenstrauss (JL) Lemma \cite{Johnson1984Extensions, dasgupta1999elementary} and the Restricted Isometry Property (RIP) \cite{Candes2005Decoding, Cand2008The, Baraniuk2015A} allow the use of random projection to reduce the space dimension while keeping the Euclidean distance between any two data points, which leads to the boom of Compressed Sensing (CS) and the field of sparsity related signal processing \cite{Donoho2006Compressed, Candes2006Robust, Aeron2010Information, candes2007sparsity, eldar2012compressed}.

Typically, the problem of CS is described as
$$
{\bf y}={\bm \Phi}{\bf x},
$$
where ${\bf x} \in \mathbb{R}^N$ is a $k$-sparse signal we wish to recover, ${\bf y} \in \mathbb{R}^n, n < N$ is the available measurement, and ${\bm \Phi} \in \mathbb{R}^{n \times N}$ is a known projection matrix. 
In order to sufficiently ensure robust recovery of the original signal, the projection matrix should approximately preserve the distance between any two $k$-sparse signals. 
Specifically, JL Lemma states that, for any set $\set V$ of $L$ points in $\mathbb{R}^N$, if $n$ is a positive integer such that
$$
n \ge 4\left(\frac{\varepsilon^2}{2} - \frac{\varepsilon^3}{3}\right)^{-1}\!\!{\rm ln} L, 
$$
then there exists a map $f : \mathbb{R}^N \to \mathbb{R}^n$, such that for all ${\bf x}_1, {\bf x}_2 \in {\set V}$,
$$
(1\!-\!\varepsilon)\|{\bf x}_1 \!-\! {\bf x}_2\|_2^2 \le \|f({\bf x}_1) \!-\! f({\bf x}_2)\|_2^2 \le (1\!+\!\varepsilon)\|{\bf x}_1 \!-\! {\bf x}_2\|_2^2, 
$$
where $0 < \varepsilon < 1$ is a constant. 
Moreover, RIP is a generalization of this lemma. 
We say that the projection matrix ${\bm \Phi}$ satisfies RIP of order $k$ with $\delta_k$ as the smallest nonnegative constant, such that
$$
(1\!-\!\delta_k)\|{\bf x}_1 \!-\! {\bf x}_2\|_2^2 \le \|{\bm \Phi}{\bf x}_1 \!-\! {\bm \Phi}{\bf x}_2\|_2^2 \le (1\!+\!\delta_k)\|{\bf x}_1 \!-\! {\bf x}_2\|_2^2.
$$
holds for any two $k$-sparse vectors ${\bf x}_1$ and ${\bf x}_2$.

In CS, we generally construct  the measurement matrix by selecting ${\bm \Phi}$ as a random matrix. 
For example, we draw the matrix elements $\phi_{ij}$ independently from Gaussian distribution $\mathcal{N}(0,1/n)$ \cite{Donoho2006Compressed, Candes2006Robust, baraniuk2006johnson}. 
More rigorously, using concentration of measure arguments \cite{levy1951problemes, achlioptas2001database, ledoux2005concentration}, ${\bm \Phi}$ is shown to have the RIP with high probability if $n \ge ck{\rm ln}(N/k)$, with $c$ a small constant. 
In addition, there are theoretical results showing some angle-preserving properties as well \cite{Haupt2007A, Chang2012Achievable}.

Furthermore, in \cite{Gedalyahu2010Time, Eldar2009Robust, wang2013fast}, the signals of interest have been extended from conventional sparse vectors to the vectors that belong to a union of subspaces (UoS). 
Nowadays, UoS becomes an important topic \cite{Eldar2009Robust}, and plays significant role in many subfields of CS, such as multiple measurement vector \cite{chen2006theoretical, Cotter2005Sparse} and block sparse recovery \cite{eldar2010block, Duarte2011Structured}. 
It has been proved in \cite{Davenport2010Signal, Blumensath2009Sampling, agarwal2007embeddings, magen2002dimensionality} that, with high probability the random projection matrix ${\bm \Phi}$ can preserve the distance between two signals belonging to a UoS. 
Recently, the stable embedding property has been extended to signals modeled as low-dimensional Riemannian submanifolds in Euclidean space \cite{Eftekhari2013New, Baraniuk2010Random, Yap2013Stable}.

\subsection{Motivation}

In the era of data deluge, labelling huge amount of large scale data can be time-consuming, costly, and sometimes intractable. 
As a consequence, unsupervised learning attracts increasing attention in recent years. 
One such method emerging recently, subspace clustering (SC) \cite{Elhamifar2009Sparse, soltanolkotabi2012geometric, elhamifar2013sparse, Heckel2015Robust}, that depicts the latent structure of a variety of data as a union of subspaces (Fig. \ref{fig:showCSC} (a) and (b)), has been shown to be powerful in a wide range of applications, including motion segmentation, face clustering, and anomaly detection. It also offers great potential for previously less explored datasets, such as network data, gene series, and medical images.

\begin{figure}
  \centering
  \subfigure[]{
    \label{fig:showCSC:a} %% label for first subfigure
    \includegraphics[width=0.3\textwidth]{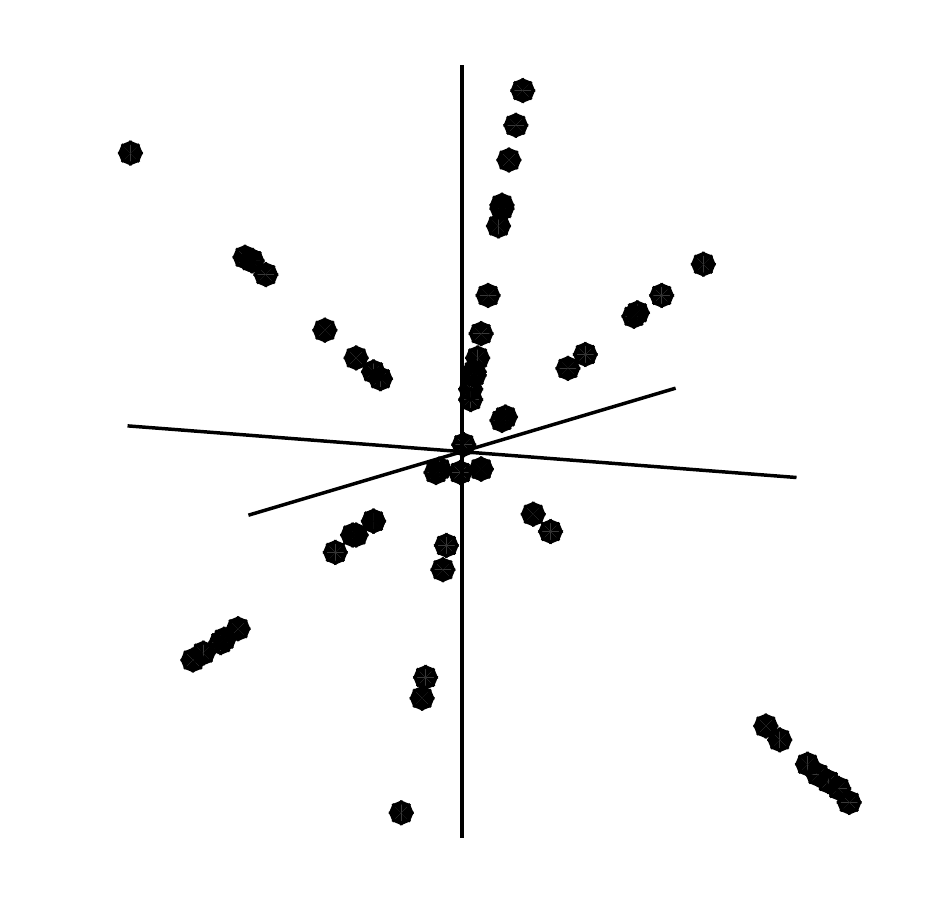}}
  \subfigure[]{
    \label{fig:showCSC:b} %% label for second subfigure
    \includegraphics[width=0.3\textwidth]{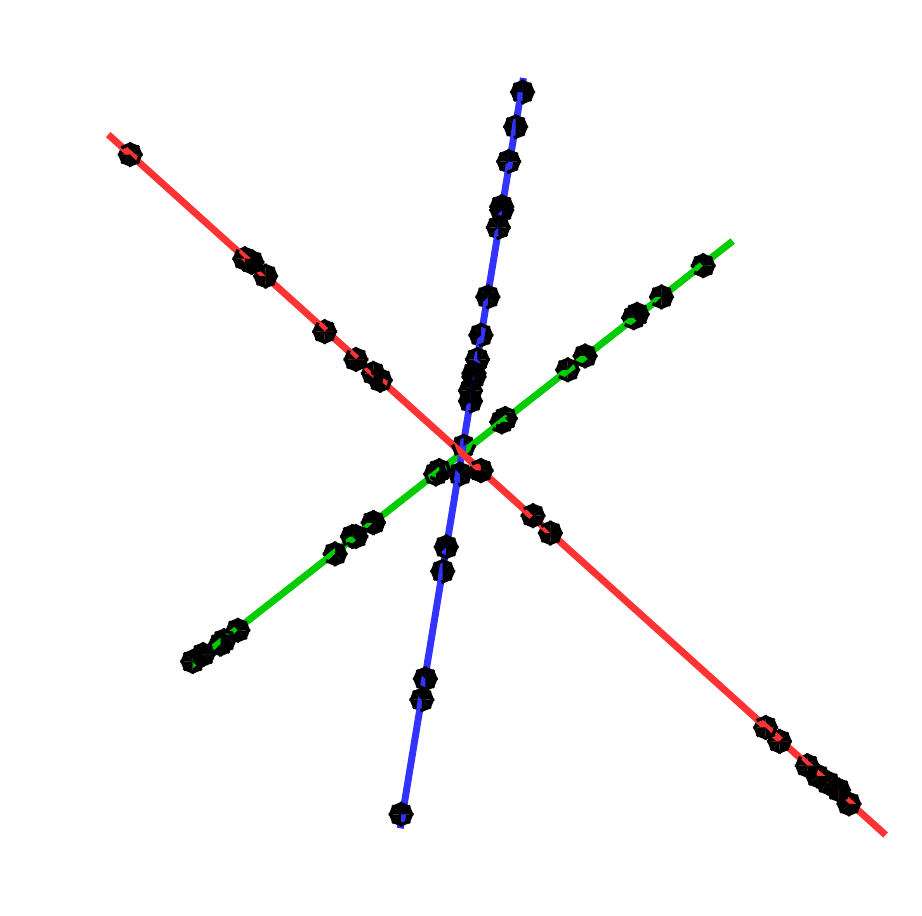}}\\
  \subfigure[]{
    \label{fig:showCSC:c} %% label for first subfigure
    \includegraphics[width=0.3\textwidth]{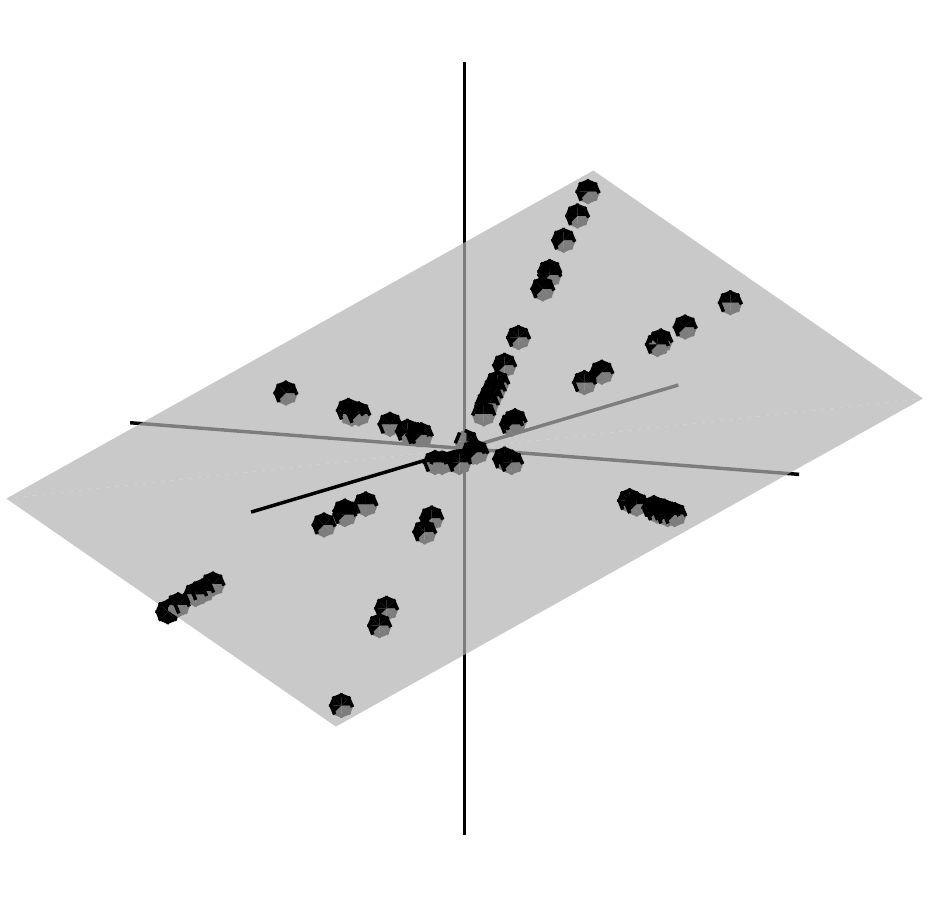}}
  \subfigure[]{
    \label{fig:showCSC:d} %% label for second subfigure
    \includegraphics[width=0.3\textwidth]{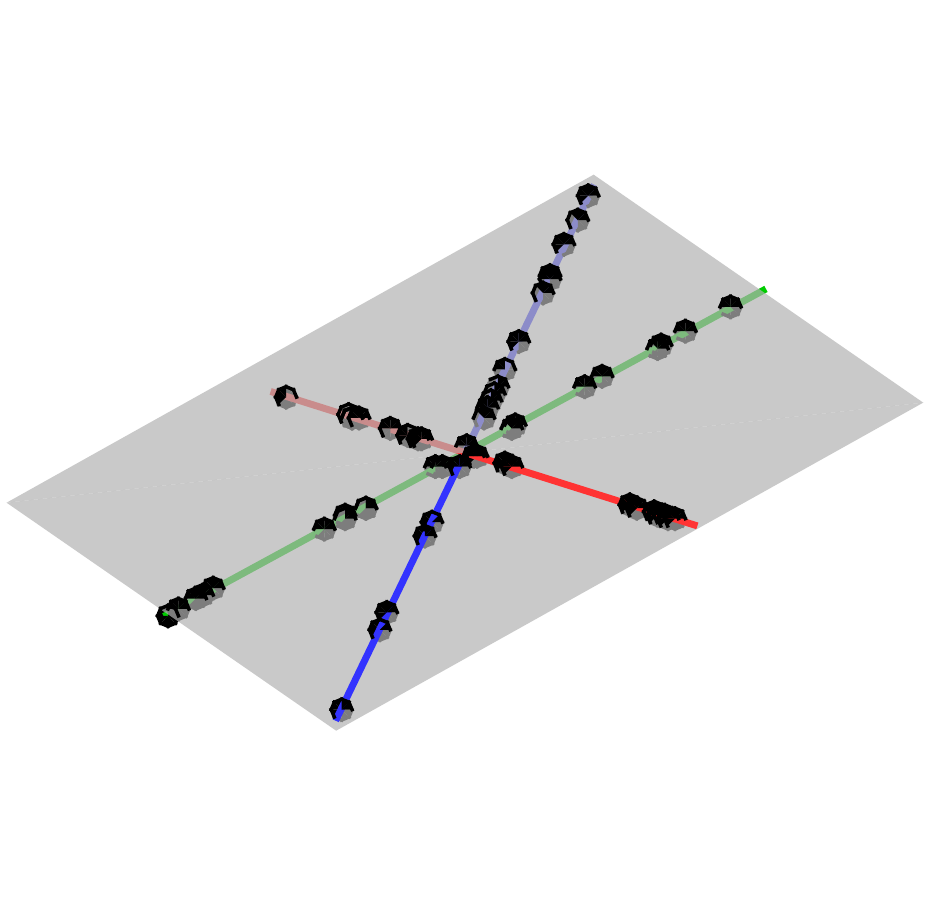}}\\
  \subfigure[]{
    \label{fig:showCSC:e} %% label for first subfigure
    \includegraphics[width=0.3\textwidth]{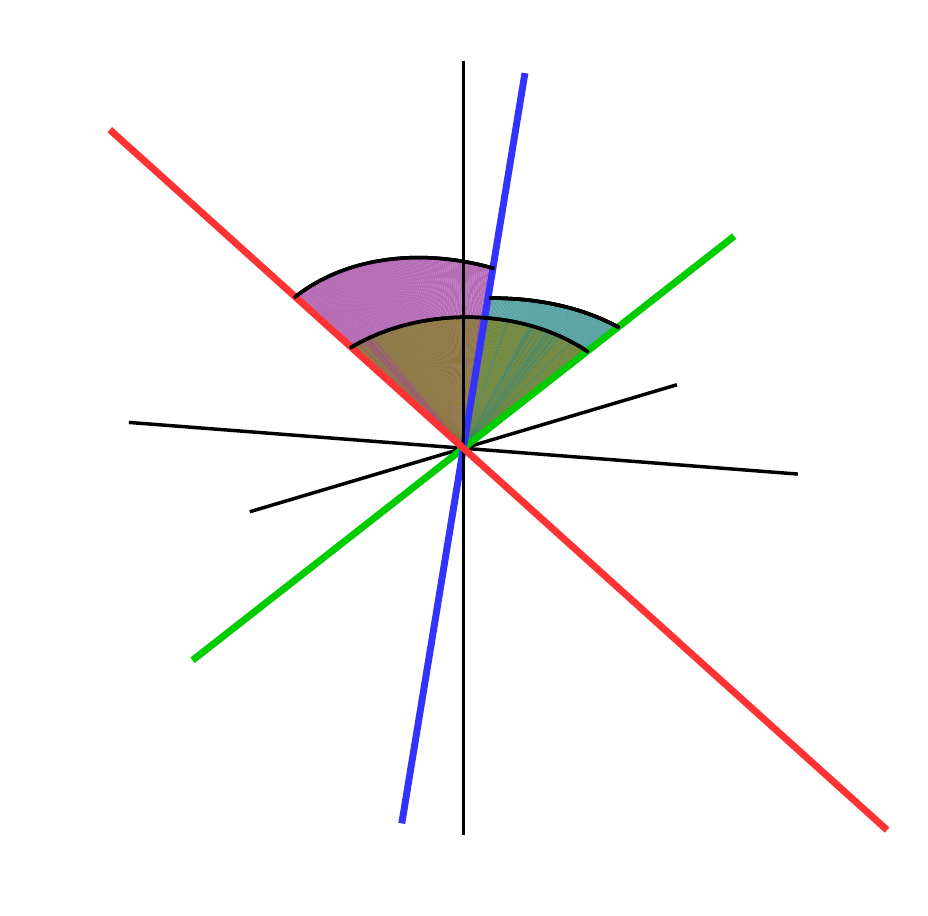}}
  \subfigure[]{
    \label{fig:showCSC:f} %% label for second subfigure
    \includegraphics[width=0.3\textwidth]{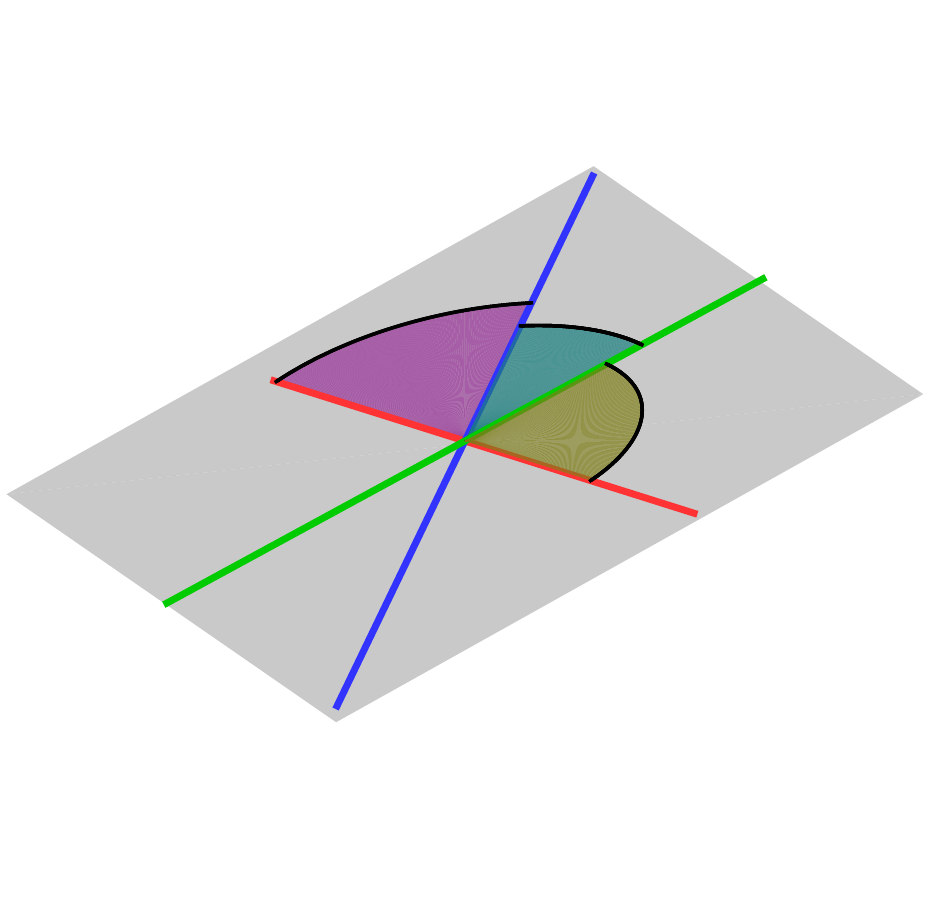}}
\caption{Visualization of SC, CSC, and the motivation of this work. 
(a) We are given unlabelled high-dimensional data in multiple classes, 
and we know the prior that data points in every same class locate in a low-dimensional latent subspace. 
(b) Our aim is to cluster the unlabelled data into several subspaces. This is called SC.
(c) In preprocessing, we reduce the dimension by using a Gaussian random matrix to project and compress the raw data into a medium-dimensional space.
(d) Then we do clustering on the compressed data. This is called CSC.
(e) The similarity or principal angles, which will be explained in Definition \ref{defi-principal-angles} and \ref{defi-affinity}, are utilized to characterize the relative position between subspaces.
(f) After random projection, the principal angles usually decrease, which may lead to clustering errors. 
In this work we will prove that the similarity between any two given subspaces is almost unchanged after random projection,
if their intrinsic dimensions are small.}
  \label{fig:showCSC} %% label for entire figure
\end{figure}

Traditional subspace clustering methods, however, suffer from the deficiency in similarity representation,
so it can be computationally expensive to adapt them to the datasets that come at a large scale. 
This leads to the revival of subspace clustering, which has become a highly active research area in machine learning, data science, and signal processing in recent years. 
To alleviate the high computational burden, how to efficiently handle large-scale datasets becomes a crucial problem, and a variety of works have been done to address this problem. 
One method considered is to perform SC on randomly compressed samples (Fig. \ref{fig:showCSC} (c) and (d)).
This is called Compressed SC (CSC) \cite{Mao2014Compressed} or dimensionality-reduced SC \cite{Heckel2014Subspace}.
Because compression reduces the dimension of ambient signal space, the computational cost on finding the self-representation in SC can be efficiently reduced. 
Based on the concept of affinity \cite{soltanolkotabi2012geometric}, which characterizes the similarity between two subspaces, the conditions under which several popular algorithms can successfully perform clustering on the compressed data have been theoretically studied and numerically verified  \cite{heckel2015dimensionality,Wang2016}.

Because the data points are randomly projected from a high-dimensional ambient space to a new medium-dimensional ambient space, a worry is that the similarity between any two low-dimensional subspaces increases
and the SC algorithms are less likely to succeed (Fig. \ref{fig:showCSC} (e) and (f)).
However, if the dimension of the latent subspace that data lie in is small compared with those of the original ambient space and the new ambient space,
we speculate whether the similarity between any two given subspaces can remain almost unchanged.
It should be highlighted that this conjecture is independent of SC algorithms.
In addition, this may benefit future studies on other subspace related topics.

Although the UoS model is very popular and is extensively used in various applications, few theoretical analysis describes the property of the random projection for linear subspaces. 
This motivates our work to discover the distance-preserving property of random projection for a set of subspaces, 
parallel to what JL Lemma guarantees for finite signal set and RIP for sparse signals.

\subsection{Main Contribution}

In this paper, motivated by the significance of JL Lemma and the feasibility of CSC, we study the RIP of Gaussian random matrices for projecting a set of finite subspaces. 
The problem is challenging as random projections neither preserve orthogonality nor normalization of the vectors defining orthonormal bases of the subspaces. 
To quantify the change in subspace affinity or distance induced by random projections both effects have to be carefully quantified.
After introducing the projection $F$-norm distance and building a metric space of the set of low-dimensional subspaces, we reveal the connection between affinity and distance, and lay a solid foundation for this work.
Then we start from a simple case that one subspace is of dimension one and prove the concentration of affinity after random projection.
Consequently, the general case is studied that the subspaces are of arbitrary dimensions.
Based on that column-wise normalization well approximates the Gram-Schmidt orthogonalization in high-dimensional scenario, we successfully reach the RIP of two subspaces. 
Finally, the main contribution is generalized to a finite set of subspaces, as stated in Theorem \ref{thm-rip-set}.

\begin{thm}\label{thm-rip-set}
For any set composed by $L$ subspaces ${\set X}_1, \cdots, {\set X}_L \in \mathbb{R}^N$ of dimension no more than $d$,
if they are projected into $\mathbb R^n$ by a Gaussian random matrix ${\bm\Phi} \in\mathbb{R}^{n\times N}$, 
$$
\set{X}_k \stackrel{\bm \Phi}{\longrightarrow} \set{Y}_k = \{{\bf y} | {\bf y}={\bm \Phi}{\bf x}, \forall {\bf x}\in \set{X}_k\},\quad k=1,\cdots,L,
$$
and $d\ll n < N$, then we have
$$
    (1-\varepsilon)D^2({\set X}_i,{\set X}_j) \le D^2({\set Y}_i,{\set Y}_j) \le (1+\varepsilon)D^2({\set X}_i,{\set X}_j), \quad \forall i, j
$$
with probability at least 
$$
	1 - \frac{2dL(L-1)}{(\varepsilon-d/{n})^2n},
$$
when $n$ is large enough, where $D(\cdot)$ denotes the projection $F$-norm distance between two subspaces.
\end{thm}

Although different metrics and distance measures have been used to describe the topological structure of the Grassmann manifold \cite{Hamm2008Grassman, Miao1992On, Qiu2005Unitarily}, as far as we know, there is no rigorous theoretical analysis for RIP regarding directly to subspaces as opposed to points on subspaces.
This paper theoretically studies this problem for the first time. Numerical simulations are also provided to validate the theoretical results.

\subsection{Related Works}

\subsubsection{Signal Processing with Compressive Measurements \cite{Davenport2010Signal}}

This paper extends the RIP to signals that are sparse or compressible with respect to a certain basis ${\bm \Psi}$, i.e., ${\bf x} = {\bm \Psi}{\bm \alpha}$, where ${\bm \Psi}$ is represented as a unitary $N \times N$ matrix and ${\bm \alpha}$ is a $k$-sparse vector.

More rigorously, the projection matrix ${\bm \Phi}$ satisfies RIP, with respect to a certain sparsity basis ${\bm \Psi}$, of order $k$ with $\delta_k$ as the smallest nonnegative constant, such that
$$
(1\!-\!\delta_k)\|{\bf x}_1 \!-\! {\bf x}_2\|_2^2 \le \|{\bm \Phi}{\bf x}_1 \!-\! {\bm \Phi}{\bf x}_2\|_2^2 \le (1\!+\!\delta_k)\|{\bf x}_1 \!-\! {\bf x}_2\|_2^2
$$
holds for any two vectors ${\bf x}_1$ and ${\bf x}_2$ that are $k$-sparse with respect to ${\bm \Psi}$. 

\subsubsection{Sampling Theorems for Signals from the Union of Linear Subspaces \cite{Blumensath2009Sampling}}

This work proves that with high probability the random projection matrix ${\bm \Phi}$ can preserve the distance between two signals belonging to a UoS.
In detail, for any $t > 0, \delta > 0$, let
$
n > \frac{2}{c\delta}\left(\ln(2L) + k\ln\left(\frac{12}{\delta}\right)+t\right),
$
where $c > 0$ is a constant. There exists a matrix ${\bm\Phi} \in \mathbb{R}^{n\times N}$ such that
$$
(1\!-\!\delta)\|{\bf x}\|_2^2 \le \|{\bm \Phi}{\bf x}\|_2^2 \le (1\!+\!\delta)\|{\bf x}\|_2^2
$$
holds for all ${\bf x}$ from the union of $L$ arbitrary $k$-dimensional subspaces.
If the entries of $\bm\Phi$ are \emph{i.i.d.} normal, this matrix satisfies the RIP with probability at least
$
1-{\rm e}^{-t}
$
and $c = \frac{7}{18}$.

\subsubsection{Embeddings of Surfaces, Curves, and Moving Points in Euclidean Space \cite{ agarwal2007embeddings}}

This paper shows that random projection preserves the structure of surfaces.
Given a collection of $L$ surfaces of linearization dimension $d$, if they are embedded into $\mathcal{O}(d\delta^{−2} \log(Ld/\delta))$ dimensions, the projected surfaces preserve their structure in the sense that for any pair of points on these surfaces the distance between them are preserved.

This paper also shows that, besides preserving pairwise distances of the moving points, random projection is able to preserve the radius of the smallest enclosing ball of the moving points at every moment of the motion.

\subsubsection{Dimensionality reductions that preserve volumes and distance to affine spaces, and their algorithmic applications \cite{magen2002dimensionality}}

The main contribution of this work is stated as follows. If $S$ is an $n$ point subset of $\mathbb{R}^N$, $0 < \delta < \frac{1}{3}$ and 
$
n = 256d\log n(\max\{d, 1/\delta\})^2,
$
there is a mapping of $\mathbb{R}^N$ into $\mathbb{R}^n$ under which volumes of sets of size at most $d$ do not change by more than a factor of $1+\delta$, and the distance of points from affine hulls of sets of size at most $k - 1$ is preserved within a relative error of $\delta$.

The above related works all study the distance preserving properties of compressed data points,
which may be sparse on specific basis or lie in a couple of subspaces or surfaces.
As far as we know, there is no work that studies the RIP for subspaces directly.
In this paper, for the first time we extend the object of RIP from a set of \emph{data points} to a set of \emph{subspaces}.
As stated in Theorem \ref{thm-rip-set}, 
we prove that with high probability \emph{the distance between two projected subspaces} is concentrated around its estimates.
When the ambient dimension after projection is sufficiently large, the distance between two subspaces almost remain unchanged after random projection.

In both related works and this paper, the same mathematical tool of concentration inequalities
is adopted to derive the RIP for two different objects, data points in Euclidean space and subspaces in Euclidean space (or points in Grassmann manifold), respectively.
Considering that both Euclidean space and random projection are linear but Grassmannian is not linear let along the projection on it,
one has to admit that the new problem is much more difficult than the existing one.
In an intuitive way, our problem is more challenging as random projections neither preserve orthogonality nor normalize of the vectors defining orthonormal bases of the subspaces.
Finally, the essential part of the work is solving the above challenges with geometric proof. 
This technique has hardly been used previously to derive the RIP for data points.

\section{Problem Formulation}

We first introduce the principal angles to describe the relative position and the affinity to measure the similarity between two subspaces.
Considering that the relation of subspaces reflected by affinity does not possess good features of metric space, we then introduce the projection $F$-norm distance for evaluating the separability of subspaces.
To be highlighted, we discover the connection between affinity and the above distance, which lays the foundation for the main contribution of this work. 

\subsection{Principal angles and affinity}

The principal angles (or canonical angles) between two subspaces provide the best way to characterize the relative subspace positions. It has been introduced by Jordan \cite{jordan1875essai} in 1875 and then rediscovered for several times. One may read \cite{PrincipalAngles2006} and the references herein for more usages of principal angles.

\begin{defi}\label{defi-principal-angles} The principal angles $\theta_1,\cdots,\theta_{d_1}$ between two subspaces ${\set X}_1$ and ${\set X}_2$ of dimensions $d_1\le d_2$, are recursively defined as
\begin{equation}
\cos{\theta_i}=\max\limits_{{\bf x}_1 \in {\set X}_1}\max\limits_{{\bf x}_2 \in {\set X}_2}\frac{{\bf x}_1^{\rm T}{\bf x}_2}{\Vert {\bf x}_1\Vert \Vert {\bf x}_2\Vert}=:\frac{{\bf x}_{1i}^{\rm T}{\bf x}_{2i}}{\Vert {\bf x}_{1i}\Vert\Vert {\bf x}_{2i}\Vert},
\end{equation}
with the orthogonality constraints ${\bf x}_{k}^{\rm T}{\bf x}_{kj}=0, j=1,\cdots,i-1,k=1,2$.
\end{defi}

An alternative way of computing principal angles is to use the singular value decomposition \cite{PrincipalAngles1973}.

\begin{lem}\label{lema-principal-angles-2}
Let the columns of ${\bf U}_k$ be orthonormal bases for subspace ${\set X}_k$ of dimension $d_k, k=1,2$ and suppose $d_1\le d_2$. Let $\lambda_1\ge\lambda_2\ge\cdots\ge\lambda_{d_1}\ge 0$ be the singular values of ${\bf U}_1^{\rm T}{\bf U}_2$, then 
$$
\cos\theta_i = \lambda_i, \quad i=1,\cdots, d_1.
$$
\end{lem}

When studying the problem of subspace clustering, affinity has been defined by utilizing principal angles to measure the subspace similarity \cite{soltanolkotabi2012geometric}.

\begin{defi}\label{defi-affinity} The affinity between two subspaces ${\set X}_1$ and ${\set X}_2$ of dimension $d_1\le d_2$ is defined as
\begin{equation}\label{eq-define-affinity-2}
\aff\left({\set X}_1, {\set X}_2\right):=\left(\sum_{i=1}^{d_1}\cos^2\theta_i\right)^{\frac12}.
\end{equation}
\end{defi}

Using Lemma \ref{lema-principal-angles-2} in Definition \ref{defi-affinity}, we may readily introduce an algebraic approach for calculating affinity as follows.

\begin{lem}\label{defi-affinity}
The affinity between two subspaces ${\set X}_1$ and ${\set X}_2$ can be calculated by
\begin{equation}\label{eq-define-affinity}
\aff(\set{X}_1, \set{X}_2) := \|{\bf U}_1^{\rm T}{\bf U}_2\|_F,
\end{equation}
where the columns of ${\bf U}_k$ are orthonormal bases of $\set{X}_k, k=1,2$. 
\end{lem}

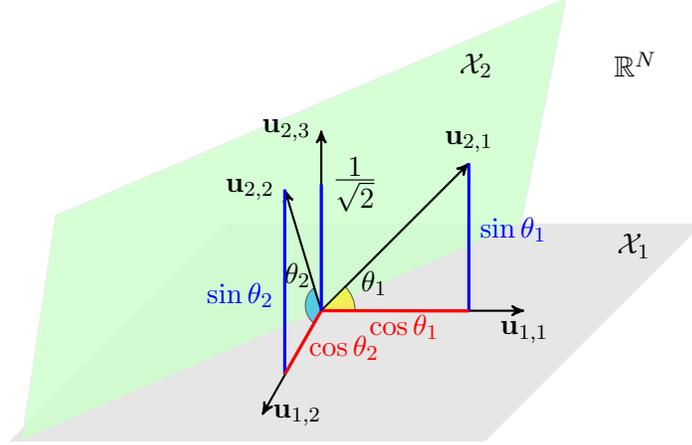
\begin{figure}[t]
\begin{center}
\begin{tikzpicture}[scale=3,
    axis/.style={ thick, ->, >=stealth'},
    important line/.style={thick},
    dashed line/.style={dashed, thin},
    pile/.style={thick, ->, >=stealth', shorten <=2pt, shorten >=2pt},
    every node/.style={color=black}
    ]
    % define origin and plot
    \coordinate (O) at (0,0,0);
    \node[left=6pt,above]  at (O)  {$O$};
    \node[above] at (1.2,0.8,-0.5) {$\mathbb{R}^N$};
    % plot X2
    \filldraw[fill=gray!20,draw=gray!20] (-0.8,0,-1) -- (-0.8,0,1.5) -- (1.3,0,1.5) -- (1.3,0,-1) -- cycle;
    \node[below] at (1,0,-1) {${\mathcal X}_1$};
    % plot X1
    \filldraw[fill=green!20,draw=gray!20,opacity=0.8] (0.5,0,-1) -- (0.7,1,-1) -- (-0.6,1,1.5) -- (-0.75,0,1.5) -- cycle;
    \node[below] at (0.3,0.8,-1) {${\mathcal X}_2$};
    % plot first principal angle
    \coordinate (v21) at (0.9,0,0);
    \draw[axis] (O) -- (v21) node [right,below] {${\bf u}_{1,1}$};
    \coordinate (v11) at (0.5,0.5,-0.4);
    \draw[axis] (O) -- (v11) node [above] {${\bf u}_{2,1}$};
    \filldraw[fill=yellow, draw=black,opacity=0.6] (0,0) -- (0.15,0.0) arc (0:45:0.15) -- cycle;
    \node[right] at (0,0,0) [right=20,above=2] {$\theta_1$};
    \coordinate (v11pY2) at (0.5,-0.155,-0.4);
    \draw[very thick,blue] (v11) node [below=25,right,blue] {$\sin\theta_1$} -- (v11pY2) ;
    \draw[very thick,red] (O) -- (v11pY2) node [below=6,left=7,red] {$\cos\theta_1$};
    % plot second principal angles    
    \coordinate (v22) at (0.2,0,1.2);
    \draw[axis] (O) -- (v22) node [right] {${\bf u}_{1,2}$};
    \coordinate (v12) at (0.3,1,1.2);
    \draw[axis] (O) -- (v12) node [left] {${\bf u}_{2,2}$};
    \filldraw[fill=cyan, draw=black,opacity=0.6] (0,0) -- (-0.035,0.1) arc (130:230:0.1) -- cycle;
    \node[right] at (0,0,0) [left=9,above=5] {$\theta_2$};
    \coordinate (v12pY2) at (0.3,0.18,1.2);
    \draw[very thick,blue] (v12) node [below=40,left,blue] {$\sin\theta_2$} -- (v12pY2) ;
    \draw[very thick,red] (O) -- (v12pY2) node [above=10,right=5,red] {$\cos\theta_2$};
    % plot the third vector
    \coordinate (v13) at (0,0.8,0);
    \coordinate (v13half) at (0,0.56,0);
    \draw[axis] (O) -- (v13) node [left] {${\bf u}_{2,3}$};
    \draw[very thick,blue] (O) -- (v13half) node [right] {$\frac{\textstyle1}{\textstyle \sqrt{2}}$};
\end{tikzpicture}
\end{center}
\caption{Visualization of principal angles, affinity, and generalized projection $F$-norm distance. 
$\set{X}_1$ and $\set{X}_2$ are two subspaces of dimension $2$ and $3$. 
$\theta_1$ and $\theta_2$ are the principal angles.
Affinity and generalized projection $F$-norm distance are denoted, respectively, as the \emph{lengths} of red bars and blue bars.
\label{fig:principalangles}}
\end{figure}

\subsection{From similarity to distance}
\label{sec-distance}

In order to take advantage of the properties of metric space, we prefer to study the relation between subspaces by distances.
Actually, there are several different definitions of distance based on the principal angles between two subspaces \cite{SeveralDistances1998}. 
In this work, we focus on the projection $F$-norm distance.
When two subspaces are of the same dimension, this distance is defined as follows.

\begin{defi}\label{defi-distance-same-d}
The projection $F$-norm distance between two subspaces ${\set X}_1$ and ${\set X}_2$ of the same dimension $d$ is defined as
\begin{align*}
D_1({\set X}_1,{\set X}_2) :=\frac1{\sqrt 2} \|{\bf P}_1 - {\bf P}_2\|_F =\bigg(\sum_{i=1}^d \sin^2 \theta_i\bigg)^{\frac12},
\end{align*}
where ${\bf P}_k = {\bf U}_k{\bf U}_k^{\rm T}$ denotes the projection matrix for subspace ${\set X}_k, k=1,2$ and $\theta_i, 1\le i\le d$ denote the principal angles between the two subspaces.
\end{defi}

When the dimensions of the two subspaces are different, which is the case of wide applications, we may accordingly generate the projection $F$-norm distance as follows.

\begin{defi}\label{defi-distance}
The generalized projection $F$-norm distance between two subspaces ${\set X}_1$ and ${\set X}_2$ of dimension $d_1, d_2$ is defined as
\begin{equation}\label{eq-define-distance}
D({\set X}_1,{\set X}_2) :=\frac1{\sqrt 2} \|{\bf P}_1 - {\bf P}_2\|_F,
\end{equation}
where ${\bf P}_k = {\bf U}_k{\bf U}_k^{\rm T}$ denotes the projection matrix for subspace ${\set X}_k, k=1,2$.
\end{defi}

One may readily check that this definition meets all requirements in the definition of distance measure, i.e. non-negativity, positive-definiteness, symmetry, and triangular inequality, thus the space of different dimensional subspaces becomes a metric space. This could also be derived by deeming ${\set X}_k$ as points on Grassmann manifold \cite{Grassmann2004}.

\begin{rem}
We want to stress that \eqref{eq-define-distance} is different from
$
 \big(\sum_{i=1}^{d_1} \sin^2 \theta_i\big)^{1/2},
$
which does not possess positive-definiteness, therefore violates the definition of a distance. 
\end{rem}

Combining Definition \ref{defi-affinity} and Definition \ref{defi-distance}, we reveal for the first time the relationship between distance and affinity.

\begin{lem}\label{lem-aff-to-dist}
The \emph{distance} and \emph{affinity} between two subspaces ${\set X}_1$ and ${\set X}_2$ of dimension $d_1, d_2$, are connected by 
\begin{equation}\label{eq-aff-to-dist}
 D^2(\set{X}_1, \set{X}_2) = \frac{d_1+d_2}{2}-\aff ^2(\set{X}_1, \set{X}_2).
\end{equation}
\end{lem}

\begin{proof}
The proof is postponed to Appendix \ref{Appendixproof-lem-aff-to-dist}.
\end{proof}

Because of the concise definition and easy computation of affinity, in this work we always start the theoretical analysis with affinity, and then present the results with distance by using Lemma \ref{lem-aff-to-dist}. 
In addition, many discussions are also conducted based on the concept of affinity. 
The relations among principal angles, affinity, and distance are visualized in Fig. \ref{fig:principalangles}.

\subsection{Projection of subspaces}

We will focus on the RIP of randomly projecting two low-dimensional subspaces from a high-dimensional ambient space to a medium-dimensional ambient space.

\begin{defi}\label{defi-project-subspace}
Let $\set{X}_1, \set{X}_2\subset \mathbb{R}^N$ be two subspaces of dimension $d_1 \le d_2\ll N$. 
They are randomly projected to $\mathbb{R}^n, d_2 \ll n<N$ as ${\mathcal Y}_k$,
$$
	\set{X}_k \stackrel{\bm \Phi}{\longrightarrow} \set{Y}_k = \{{\bf y} | {\bf y}={\bm \Phi}{\bf x}, \forall {\bf x}\in \set{X}_k\}, \quad k=1,2.
$$
where the projection matrix ${\bm \Phi}\in\mathbb{R}^{n\times N}$, $n<N$, is composed of entries independently drawn from Gaussian distribution $\mathcal{N}(0,1/n)$. 
\end{defi}

\begin{rem}
Because $d_1\le d_2 < n$, one may notice that the dimension of subspaces remains unchanged after random projection with probability one. 
\end{rem}

Following the above definition, we will study the change of the distance caused by random projection. 
For simplifying notation, we denote $\dX = D({\set X}_1,{\set X}_2)$ and $\dY = D({\set Y}_1,{\set Y}_2)$ as the distances before and after random projection.
Similarly, we use $\affX = \aff({\set X}_1,{\set X}_2)$ and $\affY = \aff({\set Y}_1,{\set Y}_2)$ to denote the affinities before and after projection.
Without loss of generality, we always suppose that $d_1\le d_2$.
We call the affinity (distance) after random projection as \emph{projected affinity} (\emph{distance}).

\section{Main Results}

In this section, we present our results about the RIP of subspaces after random projection.
It begins with a simple case of estimating the projected affinity of a line and a subspace, and then the result is extended to the case of two subspaces with arbitrary dimensions.
Finally, the RIP for subspaces is stated.\footnote{Notice that the notation of \emph{$\lesssim$} in this work holds in the sense of equivalence. 
For example, if $f(n)\le 1/(n-2)$, we may state that, without confusion, $f(n)\lesssim 1/n$ when $n$ is large enough for simplicity, considering that $1/(n-2)\sim 1/n$. 
}
It should be noted that all conclusions made in this paper are based on the assumption that $d_1\le d_2 \ll n$.

\subsection{Concentration of the affinity between a line and a subspace after random projection}

We first focus on a special case that one subspace is restricted to be a line (one-dimensional subspace). We begin from the concentration of their affinity caused by random projection and then replace the metric by the introduced distance.

The affinity between a line and a subspace will increase and concentrate on an estimate after Gaussian random projection. When the dimension of the new ambient space is large enough, the affinity almost remains unchanged after projection, as revealed in Lemma \ref{lem-line-sub}.

\begin{lem}\label{lem-line-sub}
Suppose $\set X_1, \set X_2\subset \mathbb{R}^N$ are a line and a $d$-dimension subspace, $d\ge 1$, respectively. 
Let $\lambda = \aff_{\set X}$ denote their affinity. 
If they are projected onto $\mathbb R^n, n<N,$ by a Gaussian random matrix ${\bm\Phi} \in\mathbb{R}^{n\times N}$, $\set{X}_k \stackrel{\bm \Phi}{\longrightarrow} \set{Y}_k, k=1,2$, then the affinity after projection, $\aff_{\set Y}$, can be estimated by
\begin{equation}\label{lem-line-sub-aff-change}
\oaffYS = \lambda^2+\frac{d}{n}\left(1-\lambda^2\right),
\end{equation}
where the estimation error is bounded by
\begin{equation}\label{lem-line-sub-aff-bound}
    \mathbb{P}\left(\left|\affYS-\oaffYS\right| > \lambda^2(1-\lambda^2)\varepsilon\right) \lesssim \frac{4}{\varepsilon^2n},
\end{equation}
when $n$ is large enough.
\end{lem}

\begin{proof}
The proof is postponed to Section \ref{Appendixproof-thm-sub-lem-line-sub}.
\end{proof}

As revealed in Lemma \ref{lem-line-sub}, the affinity between a line and a subspace increases after they are projected with the projection matrix specified as a Gaussian random matrix. 
Furthermore, the projected affinity can be estimated by \eqref{lem-line-sub-aff-change} with high probability. 
%This may be understood in four aspects.

\begin{rem}\label{rem-line-sub-1}
1)
It's evident that the projected affinity will increase. 
The reason is that the angle between the line and the subspace decreases largely after they are projected from a high-dimensional space to a medium-dimensional space.
2) 
The increment caused by projection is in direct ratio to the dimension of the subspace. 
This means that the principal angle between the line and the subspace after projection drops more, when the subspace is of larger dimensionality.
3) 
The increment caused by projection is in inverse ratio to the dimension of the ambient space after projection, $n$. 
When $n$ is large enough, the affinity remains almost unchanged after projection.
A visualization is that the smaller $n$, the larger the probability that the line is contained within the subspace after projection, or the smaller the angle between them.
4)
The increment of affinity is also determined by the affinity itself, which is $[0,1]$.
When the affinity is close to zero, which means the line is approximately orthogonal to the subspace, the increment caused by projection is the largest.
On the contrary, when the affinity approaches one, which means the line is almost contained within the subspace, the increment is the smallest.
\end{rem}

Now we will discuss the concentration of the projected affinity on its estimate, as shown in \eqref{lem-line-sub-aff-bound}.

\begin{rem}\label{rem-line-sub-2}
1)
The probability that the projected affinity deviates from its estimate is below a threshold, which reduces by the rate of $1/\varepsilon^2$, where $\varepsilon$ is in direct ratio to the accuracy of the estimate.
This demonstrates that the projected affinity concentrates well on its estimate.
2)
The threshold decreases to zero by the rate of $1/n$ when the ambient dimension $n$ increases.
This means that in the high-dimensional scenario, the projected affinity concentrates on its estimate with a very large probability.
Recalling Remark \ref{rem-line-sub-1}.3), one may conclude that the affinity remains unchanged with a large probability for high-dimensional problem.
3)
With a given probability, the accuracy of estimation also depends on the original affinity. 
Consequently, the estimate is exactly accurate in two situations, where the line is almost contained within or orthogonal to the subspace.
\end{rem}

Applying Lemma \ref{lem-aff-to-dist} in Lemma \ref{lem-line-sub} to replace affinity by distance, we may readily reach the concentration of distance when randomly projecting a line and a subspace. 

\begin{cor}
Let $\dX$ denote the distance between a line $\set X_1$ and a $d$-dimension subspace $\set X_2$. The distance after projection, $\dY$, can be estimated by
\begin{equation}\label{line-sub-dist-change}
\odYS = \dXS-\frac{d}{n}\bigg(\dXS-\frac{d-1}{2}\bigg).
\end{equation}
When $n$ is large enough, the estimation error is bounded by
\begin{equation}\label{line-sub-dist-bound}
    \mathbb{P}\left(\left|\dYS-\odYS\right| > \lambda^2(1-\lambda^2)\varepsilon\right) \lesssim \frac{4}{\varepsilon^2n}.
\end{equation}
\end{cor}

When evaluating the impact of projection with \emph{distance} instead of \emph{affinity}, the changes of the estimate, the concentration, and their dependences on $n, d$, and the original metric are similar with those in Remark \ref{rem-line-sub-1} and \ref{rem-line-sub-2}.

To sum up, we reveal that both affinity and distance between a line and a subspace concentrate on their estimates after random projection. 
By increasing the new ambient dimensionality, the metrics remain almost unchanged with a high probability.

\subsection{Concentration of the affinity between two subspaces after random projection}

We then study the general case of projecting two subspaces of arbitrary dimensions. 
Similar to the approach in the previous subsection, we begin with the concentration of affinity and then transform to distance.

The concentration of affinity between two arbitrary subspaces after random projection are revealed in Theorem \ref{thm-sub}.

\begin{thm}\label{thm-sub}
Suppose $\set X_1, \set X_2\subset \mathbb{R}^N$ are two subspaces with dimension $d_1 \le d_2$, respectively. 
Define
\begin{equation}\label{thm-sub-aff-change}
\oaffYS = \affXS+\frac{d_2}{n}(d_1-\affXS)
\end{equation}
to estimate the affinity between two subspaces after random projection, $\set{X}_k \stackrel{\bm \Phi}{\longrightarrow} \set{Y}_k, k=1,2$. When $n$ is large enough, the estimation error is bounded by
\begin{equation}\label{thm-sub-aff-bound}
    \mathbb{P}\left(\left|\affYS-\oaffYS\right| > \affXS\varepsilon\right) \lesssim \frac{4d_1}{\varepsilon^2n}.
\end{equation}
\end{thm}

\begin{proof}
The proof is postponed to Section \ref{Appendixproof-thm-sub}. 
\end{proof}

Recalling the discussions about projecting a line and a subspace in Remark \ref{rem-line-sub-1}, we may readily check that item 1) and 3) also hold in the situation of projecting two subspaces in Theorem \ref{thm-sub}.
We will discuss the other two items in Remark \ref{rem-sub-1}.

\begin{rem}\label{rem-sub-1}
1)
The increment of affinity caused by projection is in direct ratio to the larger dimension of two subspaces. 
This comes from the fact that each basis of the lower dimensional subspace may be deemed as a one-dimensional subspace, the affinity between which and the higher dimensional subspace is evaluated. 
2)
The increment caused by projection is in direct ratio to $(d_1-\affXS)$.
Notice that $\affXS\in[0,d_1]$.
When $\affXS$ is close to zero, which means that two subspaces are almost orthogonal to each other, the increment is the largest 
and in direct ratio to $d_1$.
When $\affXS$ is close to $d_1$, which means that the lower-dimensional subspace is almost contained within the other subspace, the increment must be the smallest among all situations. 
\end{rem}

When $d_1$ reduces to one, it's easy to check that the estimate of projected affinity \eqref{thm-sub-aff-change} in Theorem \ref{thm-sub} degenerates to \eqref{lem-line-sub-aff-change} in Lemma \ref{lem-line-sub}. 
However, it is obvious that the bound of \eqref{thm-sub-aff-bound} can not reduce to \eqref{lem-line-sub-aff-bound}. 
The reason is that the former is a loose result that undergoes much relaxation. 
Another version of Theorem \ref{thm-sub} is given in Lemma \ref{thm-sub-tight}, which is tight enough and exactly degenerates to Lemma \ref{lem-line-sub}.

\begin{lem}\label{thm-sub-tight}
Following the same conditions and notations in Theorem \ref{thm-sub}, when $n$ is large enough, the estimation error is bounded by
\begin{equation}\label{thm-sub-aff-bound-tight}
    \mathbb{P}\bigg(\left|\affYS-\oaffYS\right| > \sum_{i=1}^{d_1}\lambda_i^2(1-\lambda_i^2)\varepsilon\bigg) \lesssim \frac{4d_1}{\varepsilon^2n},
\end{equation}
where $\lambda_i = \cos\theta_i$ and $\theta_i$ denotes the principal angles between the original subspaces.
\end{lem}

\begin{proof}
The proof is postponed to Section \ref{Appendixproof-thm-sub}.
\end{proof}

We keep both Theorem \ref{thm-sub} and Lemma \ref{thm-sub-tight} deliberately as the main results, because they provide complementary usages.
While \eqref{thm-sub-aff-bound} produces a clear formulation which is ready to calculate by using the original affinity as a whole, \eqref{thm-sub-aff-bound-tight} reveals the relation between estimating accuracy and principal angles for us perceiving the intension.

Recalling Remark \ref{rem-line-sub-2} about projecting a line and a subspace, the first two items both hold in the scenario of projecting two subspaces.
We will recheck the last item.

\begin{rem}\label{rem-sub-2}
With a given probability, the accuracy of estimation depends on all principal angles, i.e., in direct ratio to the sum of all $\lambda_i^2(1-\lambda_i^2)$.
This means that when two subspaces are almost orthogonal to each other or one is contained within the other, the estimate is accurate.
The reason is as that in Remark \ref{rem-line-sub-2}.3).
In order to simplify notation and avoid using the concept of principal angles, the above bound is relaxed to $\affXS$ in \eqref{thm-sub-aff-bound}, as is accurate when two subspaces are orthogonal to each other.
The relaxation leads to concise expression by only using affinity without aware of principal angles.
Therefore, we recommend to apply the two bounds in respective situations.
\end{rem}

Finally, we want to highlight that the increment of affinity, the estimate accurancy, and the probability of deviation are determined by the dimensions of two subspaces, the original affinity, and the new ambient dimension, rather than the dimension of the original ambient space. This is also obvious.

Using Lemma \ref{lem-aff-to-dist} and Theorem \ref{thm-sub}, we may reach the concentration of distance between two subspaces after random projection.

\begin{cor}\label{cor-sub-distance}
Let $\dX$ denote the distance between subspaces $\set X_1$ and $\set X_2$ of dimensions $d_1 \le d_2$. The distance after Gaussian random projection, $\dY$, can be estimated by
\begin{equation}\label{sub-distance-estimate}
\odYS = \dXS-\frac{d_2}{n}\left(\dXS-\frac{d_2-d_1}{2}\right).
\end{equation}
When $n$ is large enough, the estimation error is bounded by
\begin{equation}\label{sub-dist-bound}
    \mathbb{P}\left(\left|\dYS-\odYS\right| > \dXS\varepsilon\right)  \lesssim \frac{4d_1}{\varepsilon^2n}.
\end{equation}
\end{cor}
\begin{proof}
The proof is postponed to Section \ref{Appendixproof-cor-sub-distance}.
\end{proof}

\subsection{Restricted Isometry Property of random projection for subspaces}

Based on the above results, we are ready to state the RIP of subspaces.
We will begin with a special case of two given subspaces and then extend to a general case of any two candidates in a finite set of subspaces.
As a consequence, we generalize the JL lemma from set of points to set of subspaces.

\begin{thm}\label{thm-rip}
Suppose $\set X_1, \set X_2\subset \mathbb{R}^N$ are two subspaces with dimension $d_1 \le d_2$, respectively. 
If $\set X_1$ and $\set X_2$ are projected into $\mathbb R^n$ by a Gaussian random matrix ${\bm\Phi} \in\mathbb{R}^{n\times N}$, $\set{X}_k \stackrel{\bm \Phi}{\longrightarrow} \set{Y}_k, k=1,2$, then we have
\begin{equation}\label{rip}
    (1-\varepsilon)\dXS \le \dYS  \le (1+\varepsilon)\dXS,
\end{equation}
with probability at least 
\begin{equation}\label{eq-rip-probability}
	1 - \frac{4d_1}{(\varepsilon-{d_2}/{n})^2n},
\end{equation}
when $n$ is large enough.
\end{thm}

\begin{proof}
The proof is postponed to Section \ref{Appendixproof-thm-rip}.
\end{proof}

Theorem \ref{thm-rip} shows that when $n$ is sufficiently large, the distance between two subspaces remains unchanged with high probability after random projection. 

According to Theorem \ref{thm-rip}, we can readily conclude the RIP of finite subspaces set in Theorem \ref{thm-rip-set}.

\begin{rem}
It should be noticed that the RIP for all low-dimensional subspaces does not hold in a way similar to sparse signals. 
Even for all one-dimensional subspaces which include all directions in the ambient space, any projection matrix will reduce some subspaces to the origin.
Therefore there is no RIP for compressing all one-dimensional subspaces in any cases.
\end{rem}

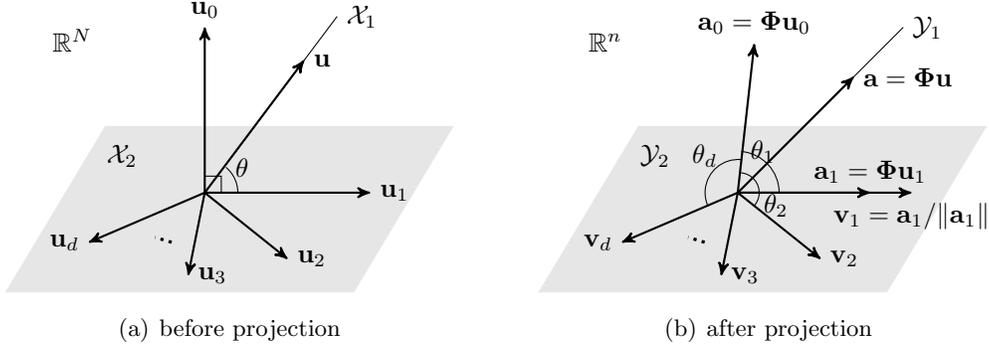
\begin{figure*}[!h]
  \centering
  \subfigure[before projection]{
    \label{fig:subfig:a}
\begin{small}
\begin{tikzpicture}[
    scale=2.2,
    axis/.style={ thick, ->, >=stealth'},
    important line/.style={thick},
    dashed line/.style={dashed, thin},
    pile/.style={thick, ->, >=stealth', shorten <=2pt, shorten >=2pt},
    every node/.style={color=black}
    ]
    % R^N
    \node[above] at (-0.8,0.8) {$\mathbb{R}^N$};
    % X2
    \filldraw[fill=gray!20,draw=gray!20] (1.5,0.4) -- (0.9,-0.6) -- (-1.2,-0.6) -- (-0.6,0.4) -- (1.5,0.4);
    \node[above] at (-0.5,0.1) {$\mathcal{X}_2$};
    % X1
    \draw[axis] (0,0) coordinate (O) -- (0.6,0.8) node(u) [right] {${\bf u}$};
    \draw (O) -- (0.8,1.066) node [right] {$\mathcal{X}_1$};
    \draw (0.2,0) arc (0:50:0.2) node [right] {$\theta$};
    % axis
    \draw[axis] (O) -- (0,1) node(u0) [above] {${\bf u}_0$};
    \draw[axis] (O) -- (1,0) node(u1) [right] {${\bf u}_1$};
    \draw (0.1,0) -- (0.1,0.1) -- (0,0.1);
    \draw[axis] (O) -- (0.5,-0.4) node(u2) [right] {${\bf u}_2$};
    \draw[axis] (O) -- (-0.1,-0.5) node(u3) [right] {${\bf u}_3$};
    \draw[axis] (O) -- (-0.7,-0.3) node(ud) [left] {${\bf u}_d$};
    \draw[dotted,very thick] (-0.2,-0.3) -- (-0.3,-0.27);
\end{tikzpicture}
\end{small}
  }
  \hspace{1em}
  \subfigure[after projection]{
    \label{fig:subfig:b} %% label for second subfigure
\begin{small}
\begin{tikzpicture}[
    scale=2.2,
    axis/.style={ thick, ->, >=stealth'},
    important line/.style={thick},
    dashed line/.style={dashed, thin},
    pile/.style={thick, ->, >=stealth', shorten <=2pt, shorten >=2pt},
    every node/.style={color=black}
    ]
    % R^n
    \node[above] at (-0.8,0.8) {$\mathbb{R}^n$};
    % Y2
    \filldraw[fill=gray!20,draw=gray!20] (1.5,0.4) -- (0.9,-0.6) -- (-1.2,-0.6) -- (-0.6,0.4) -- (1.5,0.4);
    \node[above] at (-0.5,0.1) {$\mathcal{Y}_2$};
    % Y1
    \draw[axis] (0,0) coordinate (O) -- (0.7,0.7) node(a) [right] {${\bf a}={\bm\Phi}{\bf u}$};
    \draw (O) -- (1,1) node [right] {$\mathcal{Y}_1$};
    % axis
    \draw[axis] (O) -- (0.1,0.9) node(a0) [above] {${\bf a}_0={\bm\Phi}{\bf u}_0$};
    \draw[axis] (O) -- (0.8,0) node(a1) [above] {${\bf a}_1={\bm\Phi}{\bf u}_1$};
    \draw (0.25,0) arc (0:85:0.25) node [above,right] {$\theta_1$};
    \draw[axis] (O) -- (1.05,0) node(v1) [below] {${\bf v}_1={\bf a}_1/\|{\bf a}_1\|$};
    \draw[axis] (O) -- (0.5,-0.4) node(v2) [right] {${\bf v}_2$};
    \draw (0.02,0.12) arc (85:-40:0.12) node [right] {$\theta_2$};
    \draw[axis] (O) -- (-0.1,-0.5) node(v3) [right] {${\bf v}_3$};
    \draw[axis] (O) -- (-0.7,-0.3) node(vd) [left] {${\bf v}_d$};
    \draw (0.02,0.2) arc (85:205:0.2);
    \node[above] at (-0.2,0.1) {$\theta_d$};
    \draw[dotted,very thick] (-0.2,-0.3) -- (-0.3,-0.27);
\end{tikzpicture}
\end{small}
  }
  \caption{Visualization of the notations in the proof of Lemma \ref{lem-line-sub}.}
  \label{fig:visualLemma2} %% label for entire figure
\end{figure*}

\begin{rem}
One may already notice that affinity, which measures similarity, can not provide a characteristic like RIP. 
A direct example is that for two independent subspaces with zero affinity, the projected affinity always deviates from zero.
On the contrary, for two subspaces of zero distance, which only happens in the situation that the subspaces are exactly identical, the projected distance is of course zero.
This visualizes our motivation of introducing the distance and building a metric space in Section \ref{sec-distance}.
\end{rem}

\section{Proofs of the Main Results}

Before proving the main results, we would like to define Gaussian random vector and introduce an important tool in Lemma \ref{lemma-simp} for simplifying our conclusion.

\begin{defi}
A Gaussian random vector ${\bf a} \in\mathbb{R}^n$ has i.i.d. zero-mean Gaussian entries with variance $\sigma^2 = {1}/{n}$.
\end{defi}

\begin{lem}\label{lemma-simp}
Assume that $X(n)\ge0$ is a random variable indexed by $n\in\mathbb{N}^+$ and for all positive $\varepsilon_1, \varepsilon_2$, it holds that
\begin{equation}\label{Appendixequation1}
\mathbb{P}\left(X(n) > c_1\varepsilon_1 + c_2\varepsilon_2 \big| n\right) \le p(\varepsilon_1, n) + q(\varepsilon_2, n),
\end{equation}
where $c_1$, $c_2$ are two positive constants, and $p(\cdot)$, $q(\cdot)$ are two real-valued functions. If
\begin{equation}\label{Appendixequation1.5}
\lim_{n \to \infty} \frac{q(\varepsilon_2, n)}{p(\varepsilon_1, n)} = 0,\quad \forall \varepsilon_1, \varepsilon_2 > 0,
\end{equation}
then we may simplify \eqref{Appendixequation1} by
\begin{equation}\label{Appendixequation2}
\mathbb{P}\left(X(n) > c_1\varepsilon \big| n\right) \lesssim p(\varepsilon, n),
\end{equation}
when $n$ is large enough.
\end{lem}

\begin{proof}
The proof is postponed to Appendix \ref{proof-lemma-simp}.
\end{proof}

\subsection{Proof of Lemma \ref{lem-line-sub}}
\label{Appendixproof-thm-sub-lem-line-sub}

We will first check some properties of Gaussian random vectors, which will be used in the proof.

\begin{lem}\label{L4}
Let ${\bf p}, {\bf q} \in\mathbb{R}^n$ be two Gaussian random vectors, which are dependent to each other. If $\mathbb{E}{\bf p}{\bf q}^{\rm T}=\omega{\bf I}_n/n, 0\le\omega\le1$, we have
\begin{equation}
        \mathbb{P}\left(\left|\frac{\|{\bf p}\|^2}{\|{\bf q}\|^2}-1\right|>\left(1-\omega^2\right)\varepsilon\right) \lesssim \frac{4}{\varepsilon^2n} =: P_1(\varepsilon,n), 
\end{equation}
when $n$ is large enough.
\end{lem}

\begin{proof}
The proof is postponed to Appendix \ref{proof-L4}.
\end{proof}

\begin{lem}\label{C1}
Let ${\bf u}={\bf a}/{\|{\bf a}\|}$, where ${\bf a} \in\mathbb{R}^n$ is a Gaussian random vector. Let ${\bf V}=\left[{\bf v}_1,\cdots,{\bf v}_d\right]\in\mathbb{R}^{n\times d}$ denote a given orthonormal matrix and $\theta_i$ denote the angle between $\bf u$ and ${\bf v}_i, 1\le i \le d$, then we have
\begin{equation}\label{eq-p2epsilon1}
    \mathbb{P}\left(\left|\sum_{i=1}^d \cos^2\theta_i - \frac{d}{n}\right|>\varepsilon\right) 
    < \frac{2d}{\varepsilon^2n^2} =: P_2(\varepsilon,n). 
\end{equation}
\end{lem}

\begin{proof}
The proof is postponed to Appendix \ref{proof-C1}.
\end{proof}

Let us begin the proof of Lemma \ref{lem-line-sub} by choosing the bases for the line ${\set X}_1$ and the subspace ${\set X}_2$ and then calculate the compressed affinity. One may refer to Fig. \ref{fig:visualLemma2} for a visualization.

According to the definition of affinity, $\lambda=\cos\theta$, where $\theta$ is the only principal angle between ${\set X}_1$ and ${\set X}_2$.
We use $\bf u$ and ${\bf u}_1$ to denote the basis of ${\set X}_1$ and the unit vector, which constructs the principal angle with $\bf u$.
Notice that ${\bf u}_1$ locates inside ${\set X}_2$.
Consequently, we may decompose $\bf u$ by
$$
{\bf u}=\lambda{\bf u}_1+\sqrt{1-\lambda^2}{\bf u}_0,
$$
where ${\bf u}_0$ denotes some unit vector orthogonal to ${\set X}_2$.
Based on the above definition, we choose ${\bf U}=\left[{\bf u}_1, ..., {\bf u}_d\right]$ as the basis of ${\set X}_2$. 
Notice that $\{{\bf u}_2,\cdots,{\bf u}_d\}$ could be freely chosen when the orthogonality is satisfied.

After random projection, the basis of ${\set Y}_1$ changes to
\begin{align}
{\bf a} = {\bm \Phi}{\bf u} &= \lambda{\bm \Phi}{\bf u}_1+\sqrt{1-\lambda^2}{\bm \Phi}{\bf u}_0\nonumber\\
&=\lambda{\bf a}_1+\sqrt{1-\lambda^2}{\bf a}_0,\label{eq-lem-line-sub-proof-a}
\end{align}
where ${\bf a}_1={\bm\Phi}{\bf u}_1$ and ${\bf a}_0={\bm\Phi}{\bf u}_0$ are not orthogonal to each other.
As to ${\set Y}_2$, considering that ${\bm\Phi}{\bf U}$ is not a orthonormal basis, we do orthogonalization by using Gram-Schmidt process.
Denote the orthogonalized matrix by ${\bf V} = \left[{\bf v}_1,\cdots,{\bf v}_d\right]$, the first column of which
\begin{equation}\label{eq-lem-line-sub-proof-v1}
{\bf v}_1={\bf a}_1/\|{\bf a}_1\|
\end{equation}
does not change its direction after the orthogonalization.

According to the definition of affinity in \eqref{eq-define-affinity}, we may calculate the compressed affinity by
\begin{equation}\label{eq-lem-line-sub-proof-affY2}
        \affYS =\left\|\frac{{\bf a}^{\rm T}}{\|{\bf a}\|}\bf V\right\|^2 
        = \frac1{\|{\bf a}\|^2}\sum_{i=1}^d \left({\bf a}^{\rm T}{\bf v}_i\right)^2 
\end{equation}
Splitting the summation in \eqref{eq-lem-line-sub-proof-affY2} into two parts and using \eqref{eq-lem-line-sub-proof-a}, \eqref{eq-lem-line-sub-proof-v1}, and that $\bf V$ is an orthonormal matrix, we have
\begin{align}
        \quad\affYS 
        &=\frac{1}{\|{\bf a}\|^2}\bigg(\left(\lambda{\bf a}_1^{\rm T}{\bf v}_1\!+\!\sqrt{1\!-\!\lambda^2}{\bf a}_0^{\rm T}{\bf v}_1\right)^2 
         + \sum_{i=2}^d \left(\sqrt{1\!-\!\lambda^2}{\bf a}_0^{\rm T}{\bf v}_i\right)^2\bigg) \nonumber\\
        &=\frac{1}{\|{\bf a}\|^2}\!\bigg(\lambda^2\|{\bf a}_1\|^2 + 2\lambda\sqrt{1\!-\!\lambda^2}\|{\bf a}_0\|\|{\bf a}_1\|\cos\theta_1 
         + \sum_{i=1}^d (1\!-\!\lambda^2)\|{\bf a}_0\|^2\cos^2\theta_i\bigg), \label{eq-lem-line-sub-proof-affY2-1}
\end{align}
where $\theta_i$ denote the angles between ${\bf a}_0$ and ${\bf v}_i$ for $i=1,\cdots,d$. By taking the norm on both sides of \eqref{eq-lem-line-sub-proof-a}, we write
\begin{align}
\|{\bf a}\|^2 \!&= \!\left\|\lambda{\bf a}_1+\sqrt{1-\lambda^2}{\bf a}_0\right\|^2 \nonumber\\
\!&= \!\lambda^2\|{\bf a}_1\|^2 \!+\! 2\lambda\sqrt{1\!-\!\lambda^2}\|{\bf a}_0\|\|{\bf a}_1\|\cos\theta_1 
+ (1\!-\!\lambda^2)\|{\bf a}_0\|^2. \label{eq-lem-line-sub-proof-a2}
\end{align}
Eliminating $\|{\bf a}_1\|$ by inserting \eqref{eq-lem-line-sub-proof-a2} into \eqref{eq-lem-line-sub-proof-affY2-1}, we get
\begin{align}
        \affYS &=\!\frac{1}{\|{\bf a}\|^2}\!\bigg(\!\|{\bf a}\|^2 \!- \!(1\!-\!\lambda^2)\|{\bf a}_0\|^2 
        \!+\! \sum_{i=1}^d (1\!-\!\lambda^2)\|{\bf a}_0\|^2\!\cos^2\!\theta_i\bigg) \nonumber\\
        &=1 - (1-\lambda^2)\frac{\|{\bf a}_0\|^2}{\|{\bf a}\|^2}\bigg(1 - \sum_{i=1}^d \cos^2\theta_i\bigg).\label{eq-lem-line-sub-proof-affY-2}
\end{align}

We are ready for estimating ${\|{\bf a}_0\|^2}/{\|{\bf a}\|^2}$ and $\sum_{i=1}^d \cos^2\theta_i$ by using Lemma \ref{L4} and Lemma \ref{C1}, respectively. First recalling Lemma \ref{L4}, let ${\bf p} = {\bf a}_0$ and ${\bf q} = {\bf a}$. Using \eqref{eq-lem-line-sub-proof-a} we have
$$
\mathbb{E}{\bf a}_0{\bf a}^{\rm T}=\sqrt{1-\lambda^2}\mathbb{E}{\bf a}_0{\bf a}_0^{\rm T}={\sqrt{1-\lambda^2}}{\bf I}_n/{n}.
$$
Denoting $\omega=\sqrt{1-\lambda^2}$ and applying Lemma \ref{L4}, we have
\begin{equation}\label{eq-lem-line-sub-proof-P3}
\mathbb{P}\left(\left|\frac{\|{\bf a}_0\|^2}{\|{\bf a}\|^2}-1\right|>\lambda^2\varepsilon_1\right) \lesssim \frac{4}{\varepsilon_1^2n} = P_1(\varepsilon_1,n).
\end{equation}
Then recalling Lemma \ref{C1} and that ${\bf a}_0$ is independent with ${\bf V}$ by using the properties of Gaussian random distribution, we have
\begin{equation}\label{eq-lem-line-sub-proof-P4}
\mathbb{P}\left(\left|\sum_{i=1}^d \cos^2\theta_i - \frac{d}{n}\right|>\varepsilon_2\right)
< \frac{2d}{\varepsilon_2^2n^2} =  P_2(\varepsilon_2,n).
\end{equation}
Consequently, combing \eqref{eq-lem-line-sub-proof-affY-2} and \eqref{lem-line-sub-aff-change}, the estimate error is rewritten as
\begin{align}
\left|\affYS-\oaffYS\right| 
= &\left|\affYS-\left(\lambda^2+\frac{d}{n}(1-\lambda^2)\right)\right| \nonumber\\
=& \left|1 \!-\! (1\!-\!\lambda^2)\frac{\|{\bf a}_0\|^2}{\|{\bf a}\|^2}\!\!\bigg(\!1 \!-\! \sum_{i=1}^d \cos^2\!\theta_i\!\bigg)\!-\!\left(\!1\!-\!(1\!-\!\lambda^2)\!\left(1\!-\!\frac{d}{n}\right)\!\!\right)\right| \nonumber\\
=& (1-\lambda^2)\left|\left(1-\frac{d}{n}\right)-\frac{\|{\bf a}_0\|^2}{\|{\bf a}\|^2}\bigg(1 - \sum_{i=1}^d \cos^2\theta_i\bigg)\right|. \label{eq-lem-line-sub-proof-esterr}
\end{align}
By using \eqref{eq-lem-line-sub-proof-P3} and \eqref{eq-lem-line-sub-proof-P4}, the second item in RHS of \eqref{eq-lem-line-sub-proof-esterr} is bounded by
\begin{align}
\left|\left(1-\frac{d}{n}\right)-\frac{\|{\bf a}_0\|^2}{\|{\bf a}\|^2}\bigg(1 - \sum_{i=1}^d \cos^2\theta_i\bigg)\right| 
\le& \left|\left(1-\frac{d}{n}\right)-(1+\lambda^2\varepsilon_1)\left(\left(1-\frac{d}n\right)+\varepsilon_2\right)\right| \nonumber\\
\le& \left(1-\frac{d}{n}\right)\lambda^2\varepsilon_1 + \varepsilon_2 + \lambda^2\varepsilon_1\varepsilon_2 \label{eq-lem-line-sub-proof-esterrRHS-pre} \\
\sim& \left(1-\frac{d}{n}\right)\lambda^2\varepsilon_1 + \varepsilon_2, \label{eq-lem-line-sub-proof-esterrRHS}
\end{align}
with probability at least $1-P_1(\varepsilon_1,n) - P_2(\varepsilon_2,n)$, where the last item in \eqref{eq-lem-line-sub-proof-esterrRHS-pre} is dropped.

Next we will use Lemma \ref{lemma-simp} to simplify \eqref{eq-lem-line-sub-proof-esterrRHS}. Recalling that $d$ is much smaller than $n$, we have 
$$
\lim_{n\rightarrow\infty}\frac{P_2(\varepsilon_2,n)}{P_1(\varepsilon_1,n)} = \frac{d\varepsilon_1^2}{n\varepsilon_2^2} = 0.
$$
Finally, inserting \eqref{eq-lem-line-sub-proof-esterrRHS} in \eqref{eq-lem-line-sub-proof-esterr} and using Lemma \ref{lemma-simp}, we have
\begin{align}\label{eq-lem-line-sub-proof-final}
        \mathbb{P}\left(\left|\affYS-\oaff^2\right|>(1\!-\!\lambda^2)\left(1\!-\!\frac{d}n\right)\lambda^2\varepsilon\right) 
        \lesssim P_1(\varepsilon,n) = \frac{4}{\varepsilon^2n}.
\end{align}
Using $d\ll n$ again to drop $\left(1-{d}/n\right)$ from \eqref{eq-lem-line-sub-proof-final}, we complete the proof.

\subsection{Proof of Lemma \ref{thm-sub-tight} and Theorem \ref{thm-sub}}
\label{Appendixproof-thm-sub}

We introduce two major techniques utilized in this proof.

\subsubsection{\emph{Quasi}-orthonormal basis}

We introduce a \emph{quasi}-orthonormal basis of the original lower dimensional subspace for estimating the projected affinity. 
In order to calculate the affinity, recalling its definition in \eqref{eq-define-affinity}, we need to prepare the orthonormal basis for both subspaces. 
However, a reasonable way is to utilize the normalized data matrix to approximate its Gram-Schmidt orthogonalization. 

\begin{figure*}[!ht]
  \centering
  \begin{small}
  \subfigure[]{
    \label{fig:realworlddata:a} %% label for first subfigure
\begin{tikzpicture}[scale=3.2,
    axis/.style={ thick, ->, >=stealth'},
    important line/.style={thick},
    dashed line/.style={dashed, thin},
    pile/.style={thick, ->, >=stealth', shorten <=2pt, shorten >=2pt},
    every node/.style={color=black}
    ]
    % define origin and plot
    \coordinate (O) at (0,0,0);
    \node[left=6pt,above]  at (O)  {$O$};
    % define vector cordinate
    \coordinate (v1) at (0,0,1);
    \coordinate (v2) at (1,0,0);
    \coordinate (a2) at (0.9,0,-0.3);
    \coordinate (v3) at (0,0.94,0);
    \coordinate (a3) at (0.25,0.95,0.3);    
    \coordinate (a3pv23) at (0.25,0.95,0);     
    \coordinate (a3pa1) at (0,0,0.3);
    \coordinate (a3pa2) at (0.9*0.24,0,-0.3*0.24);
    % plot v
    \draw[axis] (O) -- (v1) node [left] {$\bar{\bf a}_1 = {\bf v}_1$};
    
    \draw[axis] (O) -- (v2) node [right] {${\bf v}_2$};
    \draw[axis] (O) -- (a2) node [right] {$\bar{\bf a}_2$};
    \draw[axis,blue] (a2) -- (v2);
    \node[below=5,blue] at (v2)  {$-\bar{\bf a}_1\bar{r}_{1,2}$};
    
    \draw[axis] (O) -- (v3) node [above] {${\bf v}_3$};
    \draw[axis] (O) -- (a3) node [right=9,below] {$\bar{\bf a}_3$};
    \draw[dashed] (a3pa1) -- (a3);
    
    \draw[axis,blue] (O) -- (a3pa1) node [below=10,right,blue] {$\bar{\bf a}_1\bar{r}_{1,3}$};    
    \draw[dashed] (O) -- (a3pv23);
    \draw[axis,blue] (a3) -- (a3pv23) node [right=20,below,blue] {$-\bar{\bf a}_1\bar{r}_{1,3}$};
    
    \draw[axis,blue] (O) -- (a3pa2) node [above=10,right,blue] {$\bar{\bf a}_2\bar{r}_{2,3}$};    
    \draw[dashed] (a3) -- (a3pa2);    
    
    \draw[dashed] (a3pv23) -- (a3pa2);
    \draw[axis,blue] (a3pv23) -- (v3) node [above=10,right=5,blue] {$-\bar{\bf a}_2\bar{r}_{2,3}$};
    
    \draw[axis,blue] (a3) -- (v3);

    \node[below=3,left,blue] (notation) at (v3) {$-\bar{\bf a}_1\bar{r}_{1,3}-\bar{\bf a}_2\bar{r}_{2,3}$};
\end{tikzpicture}}
\subfigure[]{
    \label{fig:realworlddata:b} %% label for second subfigure
\begin{tikzpicture}[scale=3.2,
    axis/.style={ thick, ->, >=stealth'},
    important line/.style={thick},
    dashed line/.style={dashed, thin},
    pile/.style={thick, ->, >=stealth', shorten <=2pt, shorten >=2pt},
    every node/.style={color=black}
    ]
    % define origin and plot
    \coordinate (O) at (0,0,0);
    \node[left=6pt,above]  at (O)  {$O$};
    \node[above] at (1.2,0.8,-0.5) {$\mathbb{R}^n$};
    % plot Y2
    \filldraw[fill=gray!20,draw=gray!20] (-1,0,-1) -- (-1,0,1.5) -- (1.5,0,1.5) -- (1.5,0,-1) -- cycle;
    \node[below] at (1,0,-1) {${\mathcal Y}_2$};
    % define vector cordinate
    \coordinate (v21) at (0.95,0,-0.2);
    \coordinate (v22) at (0.8,0,0.6);
    \coordinate (v23) at (0.5,0,0.85);
    \coordinate (v2d) at (0,0,0.95);
    \coordinate (v11) at (0.5,0.7,-0.4);
    \coordinate (v11pY2) at (0.5,0,-0.4);
        
    \coordinate (v12) at (0.2,0.75,0.4);
    \coordinate (a12) at (0.3,0.8,0.15);
    \coordinate (a12pY2) at (0.3,0,0.15);
    
    \coordinate (v1d) at (-0.1,0.7,0.8);
    \coordinate (a1d) at (-0.1,0.8,0.6);
    \coordinate (a1dpY2) at (-0.1,0,0.6);
    % plot Y1
    \filldraw[fill=green!20,draw=gray!20,opacity=0.8] (0.5,0,-1) -- (0.7,1,-1) -- (-0.6,1,1.5) -- (-0.75,0,1.5) -- cycle;
    \node[below] at (0.3,0.8,-1) {${\mathcal Y}_1$};    
    % plot v2
    \draw[axis] (O) -- (v21) node [right,below] {${\bf v}_{2,1}$};
    \draw[axis] (O) -- (v22) node [right] {${\bf v}_{2,2}$};
    \draw[axis] (O) -- (v23) node [below] {${\bf v}_{2,3}$};
    \draw[axis] (O) -- (v2d) node [left] {${\bf v}_{2,d_2}$};
    \draw[dotted,very thick] (0.18,0,0.85) -- (0.25,0,0.83);
    % plot v1
    \draw[axis] (O) -- (v11) node [above] {$\bar{\bf a}_{1,1}={\bf v}_{1,1}$};
    \draw[dashed] (v11) -- (v11pY2) ;
    \draw[thick,red] (O) -- (v11pY2) node [right,red] {${\rm aff}_{{\mathcal Y}_1}^2$};
    
    \draw[axis] (O) -- (v12) node [left] {${\bf v}_{1,2}$};
    \draw[axis] (O) -- (a12) node [above] {$\bar{\bf a}_{1,2}$};
    \draw[dashed] (a12) -- (a12pY2) ;
    \draw[thick,red] (O) -- (a12pY2) node [right,red] {${\rm aff}_{{\mathcal Y}_2}^2$};
    \draw[axis,blue] (a12) -- (v12) node [above=15,left=-10,blue] {$-\bar{\bf a}_{1,2}r_{1,12}$};

    \draw[axis] (O) -- (v1d) node [left] {${\bf v}_{1,{d_1}}$};
    \draw[axis] (O) -- (a1d) node [above] {$\bar{\bf a}_{1,{d_1}}$};
    \draw[dashed] (a1d) -- (a1dpY2) ;
    \draw[thick,red] (O) -- (a1dpY2) node [right=10,red] {${\rm aff}_{{\mathcal Y}_{d_1}}^2$};
    \draw[axis,blue] (a1d) -- (v1d) node [left=40,above=5,blue] {$-\sum_{i=1}^{d_1-1}\bar{\bf a}_{1,i}r_{1,id_1}$};
\end{tikzpicture}}
\end{small}
\caption{(a) Visualization of Lemma \ref{lem-math}. $\{{\bf v}_i\}$ denote the Gram-Schmidt orthogonalization of a set of column-normalized vectors $\{\bar{\bf a}_i\}$. Blue vectors denote the approximation error. (b) Proof sketch of Lemma \ref{thm-sub-tight}. Blue vectors denote the estimation error in Step 2). Red bars plot the affinities in Step 3), i.e., those between $d_1$ independent one-dimensional subspaces and $\set{Y}_2$. \label{visprof}}
\end{figure*}
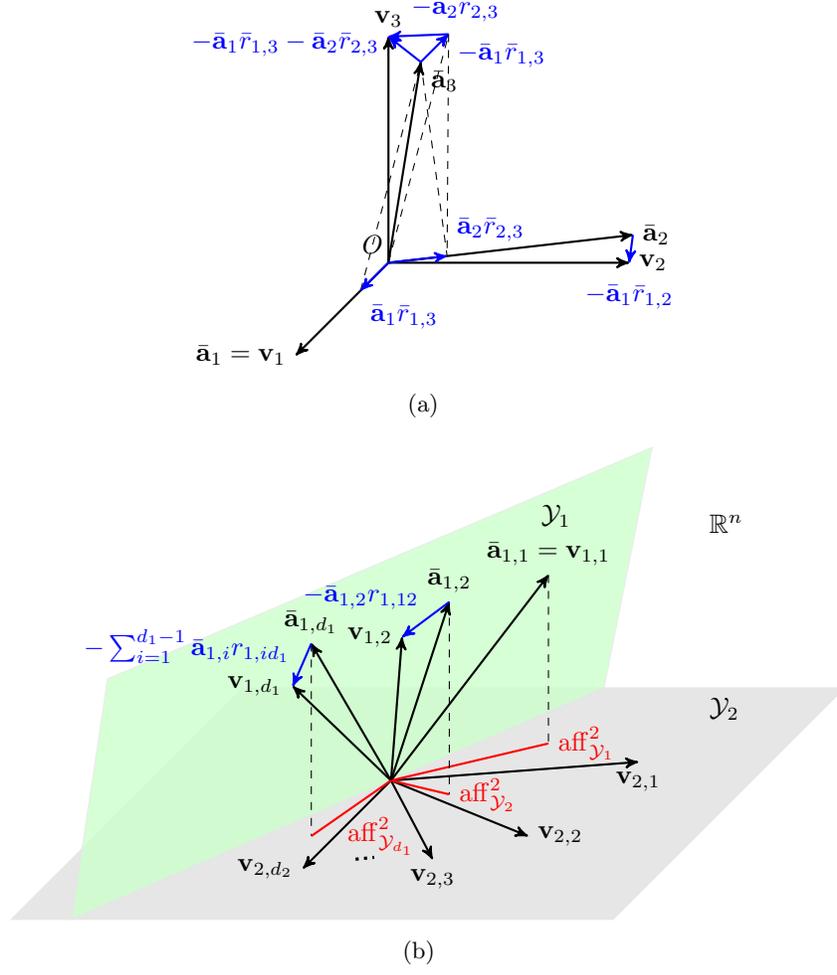

\begin{lem}\label{lem-math}
Let ${\bf V}=\left[{\bf v}_1, {\bf v}_2, \dots, {\bf v}_d\right]$ denote
the Gram-Schmidt orthogonalization
%orthonormal matrix 
of a column-normalized matrix $\bar{\bf A}=\left[\bar{\bf a}_1, \bar{\bf a}_2, \dots, \bar{\bf a}_d\right]$, 
where $\|\bar{\bf a}_i\|=1, \forall i$. %, gotten through the Gram-Schmidt orthogonalization. 
Denote $\bar{\bf R} = (\bar r_{ji}) = \bar{\bf A}^{\rm T}\bar{\bf A} - {\bf I}$. 
When $\bar{r}_{ji} = \bar{\bf a}_j^{\rm T}\bar{\bf a}_i$ is small enough for $j\ne i$, we can use $\bar{\bf A}$ to approximate ${\bf V}$ with error ${\bf V}-\bar{\bf A}=\bar{\bf A}\bar{\bf U}$, where $\bar{\bf U}=[\bar u_{ji}]\in\mathbb{R}^{d\times d}$ is an upper triangular matrix satisfying
\begin{equation}\label{eq_L1_0_1}
\bar u_{ii} = \bar g_{ii}(\bar{\bf R})\|\bar{\bf R}\|_F^2,\quad \forall i,
\end{equation}
where $\bar g_{ii}(\bar{\bf R}) > 0$ and $\lim_{\bar{\bf R} \to {\bf 0}} \bar g_{ii}(\bar{\bf R}) \le {1}/{4}$, and
\begin{equation}\label{eq_L1_0_2}
\bar u_{ji} = - \bar r_{ji} + \bar g_{ji}(\bar{\bf R})\|\bar{\bf R}\|_F, \quad \forall j<i,
\end{equation}
where $\lim_{\bar{\bf R} \to {\bf 0}} \bar g_{ji}(\bar {\bf R}) = 0$.
\end{lem}

\begin{proof}
The proof is postponed to Appendix \ref{proof-lem-math}.
\end{proof}

Lemma \ref{lem-math} unveils that, when $\bar{\bf A}$ approaches an orthonormal basis,
i.e., $\bar{\bf R}$ approaches $\bf 0$, the diagonal elements of $\bar{\bf U}$ go to zero,
and they are of the same order as $\|\bar{\bf R}\|_F^2$.
At the same time, the off-diagonal elements go to $-\bar{\bf R}$,
and the differences are of a higher order than $\|\bar{\bf R}\|_F$.
This lemma is visualized in Fig. \ref{visprof} (a).

\begin{cor}\label{C2}
Follow the definition of Lemma \ref{lem-math} and assume $\bar{r}_{ji} < \varepsilon, \forall j\ne i$. When $\varepsilon$ is small enough, for an arbitrary matrix ${\bf W} \in\mathbb{R}^{n \times l}$, we conclude that 
$$
\left|\left\|{\bf V}^{\rm T}{\bf W}\right\|_F^2-\left\|\bar{\bf A}^{\rm T}{\bf W}\right\|_F^2\right| \lesssim d\left\|\bar{\bf A}^{\rm T}{\bf W}\right\|_F^2\varepsilon.
$$
\end{cor}

\begin{proof}
Denote ${\bf W} = \left[{\bf w}_1,\cdots,{\bf w}_l\right]\in\mathbb{R}^{n\times l}$. According to Lemma \ref{lem-math}, when $\varepsilon$ is small enough, we have
\begin{align*}
\left|\left\|{\bf V}^{\rm T}{\bf W}\right\|_F^2\!-\!\left\|\bar{\bf A}^{\rm T}{\bf W}\right\|_F^2\right| &\le \sum_{i = 1}^{l} \left|\left\|{\bf V}^{\rm T}{\bf w}_{i}\right\|_2^2\!-\!\left\|\bar{\bf A}^{\rm T}{\bf w}_{i}\right\|_2^2\right| \\
&\lesssim d\sum_{i = 1}^{l} \left\|\bar{\bf A}^{\rm T}{\bf w}_{i}\right\|_2^2\varepsilon \\
&= d\left\|\bar{\bf A}^{\rm T}{\bf W}\right\|_F^2\varepsilon.
\end{align*}
%The proof is complete.
\end{proof}

Let ${\bf U}_k$ denote the orthonormal basis of ${\set X}_k$. 
After projection, the basis matrix becomes ${\bf A}_k={\bm\Phi}{\bf U}_k$, whose columns are not orthogonal to each other.
Instead of processing ${\bf A}_k$ by Gram-Schmidt process to yield the accurate orthogonal basis ${\bf V}_k$,
according to Corollary \ref{C2}, we can normalize each column of ${\bf A}_k$ to produce $\bar{\bf A}_k$ as a rather good estimate of ${\bf V}_k$, and then use $\bar{\bf A}_1^{\rm T}{\bf V}_2$ as ${\bf V}_1^{\rm T}{\bf V}_2$ to estimate the projected affinity.
We will use this technique in Step 2) of this proof.

\subsubsection{Property of Gaussian random matrix}

We need to introduce more properties of Gaussian random matrix to fulfill the proof.

\begin{lem}\label{L3} \emph{\cite{levy1951problemes, achlioptas2001database, ledoux2005concentration}}
Let ${\bf a}_1, {\bf a}_2 \in\mathbb{R}^n$ be two independent Gaussian random vectors. Let $\theta$ denote the angle between ${\bf a}_1$ and ${\bf a}_2$, then we have
\begin{equation}\label{eq-p3epsilon}
    \mathbb{P}(|\cos\theta|>\varepsilon)\le \exp\left({-\frac{\varepsilon^2n}{2}}\right) =: P_3(\varepsilon,n).
\end{equation}
\end{lem}

\begin{proof}
Equation \eqref{eq-p3epsilon} is verified by using the concentration of measure.
\end{proof}

Before going into the proof of Lemma \ref{thm-sub-tight} and Theorem \ref{thm-sub}, we first simplify the notation. For any matrix $(\cdot)_k$ or its column vector $(\cdot)_{k,i}$ in this subsection, the subscript $k$ denotes the index of subspaces, i.e., $k=1,2$. 

There are four steps in this proof. We will prepare the basis matrices for the subspaces before and after projection in the first step. 
Then following a proof sketch, Lemma \ref{thm-sub-tight} will be justified in the last three steps. Finally Theorem \ref{thm-sub} is reached by relaxing some conditions.

\emph{\bf Step 1)}
Let $\tilde{\bf U}_k=\left[\tilde{\bf u}_{k,1}, \cdots, \tilde{\bf u}_{k, d_k}\right]$ denote any orthonormal matrix for subspace ${\set X}_k$. 
According to the definition of affinity in \eqref{eq-define-affinity}, one may do singular value decomposition to $\tilde{\bf U}_2^{\rm T}\tilde{\bf U}_2$,
$$
\tilde{\bf U}_1^{\rm T}\tilde{\bf U}_2 = {\bf Q}_1{\bm\Lambda}{\bf Q}_2^{\rm T},
$$
where ${\bf Q}_1$ and ${\bf Q}_2^{\rm T}$ denote, respectively, the orthonormal basis of the column space and row space for $\tilde{\bf U}_1^{\rm T}\tilde{\bf U}_2$. 
The singular values $\lambda_i=\cos\theta_i, 1\le i\le d_1$ is located on the diagonal of $\bm\Lambda$, where $\theta_i$ denotes the $i$th principal angle. After reshaping, we have
\begin{align*}
\left(\tilde{\bf U}_1{\bf Q}_1\right)^{\rm T}\tilde{\bf U}_2{\bf Q}_2 =
{\bf U}_1^{\rm T}{\bf U}_2 = {\bm\Lambda} 
= \left[ \begin{array}{ccc|c}
     \lambda_1 & & & 0 \\ & \ddots& & \vdots \\ & & \lambda_{d_1} & 0 \end{array} \right],
\end{align*}
where ${\bf U}_k = \tilde{\bf U}_k{\bf Q}_k$ are the orthonormal basis, which has the closest connection with the affinity between these two subspaces.
Specifically, the angles of the first $d_1$ columns are all principal angles for calculating affinity, i.e.
\begin{equation}\label{eq-proof-thm-sub-orig-line-aff}
{\bf u}_{1,i}^{\rm T}{\bf u}_{2,j} = \left\{\begin{array}{ll} 
\lambda_i, & 1\le i=j \le d_1;\\
0, & {\rm elsewhere}.
\end{array}\right.
\end{equation}

After projection by using Gaussian random matrix ${\bm\Phi}$, the original basis matrix changes to ${\bf A}_k={\bm\Phi}{\bf U}_k=\left[{\bf a}_{k,1}, \cdots, {\bf a}_{k,d_k}\right]$, whose columns are not orthogonal to each other.
Considering the angles between any two columns are not very large, however, we may normalize each columns as
$$
\bar{\bf A}_k=\left[\bar{\bf a}_{k,1}, \cdots, \bar{\bf a}_{k,d_k}\right]=\left[\frac{{\bf a}_{k,1}}{\|{\bf a}_{k,1}\|}, \cdots, \frac{{\bf a}_{k,d_k}}{\|{\bf a}_{k,d_k}\|}\right],
$$
which could be used to approximate the orthonormal basis of ${\set Y}_k$.
In order to obtain the accurate orthonormal basis for the compressed subspace, the efficient method is to process $\bar{\bf A}_k$ by using Gram-Schmidt orthogonalization. We use ${\bf V}_k=\left[{\bf v}_{k,1}, \cdots, {\bf v}_{k,d_k}\right]$ to denote the orthonormal basis after orthogonalization.

Now we are ready to introduce the sketch of our proof. 
In Step 2), we use $\bar{\bf A}_1^{\rm T}{\bf V}_2$ to estimate the compressed affinity according to Corollary \ref{C2}.
In Step 3), we first deem the original subspace ${\set X}_1$ as $d_1$ independent one-dimensional subspaces and then calculate the distance between  $\bar{\bf A}_1^{\rm T}{\bf V}_2$ and the estimator of \eqref{thm-sub-aff-change} by using Lemma \ref{lem-line-sub}.
Finally, we combine the results in the above two steps and simplify it to complete the proof in the last step.
The proof sketch is visualized in Fig. \ref{visprof} (b).

\emph{\bf Step 2)}
According to the properties of Gaussian matrix, we know that $\bar{\bf a}_{1,i}$ and $\bar{\bf a}_{1,j}$, which are obtained by random projection and normalization, are independent for all $i\ne j$.
Using Lemma \ref{L3} and as a consequence, with probability at least $1-\left(d_1(d_1-1)/2\right)P_3(\varepsilon_3,n)$, we have
$$
\left|\bar{\bf a}_{1,i}^{\rm T}\bar{\bf a}_{1,j}\right|\le\varepsilon_3, \quad\forall 1\le i \ne j \le d_1.
$$
This means that $\bar{\bf A}_1$ well approximates ${\bf V}_1$ and can be roughly utilized as an orthonormal basis.
Recalling the definition of affinity in \eqref{eq-define-affinity} and using Corollary \ref{C2},  we have
\begin{align}
\left|\affYS-\left\|\bar{\bf A}_1^{\rm T}{\bf V}_2\right\|_F^2\right| &= \left|\|{\bf V}_1^{\rm T}{\bf V}_2\|_F^2-\left\|\bar{\bf A}_1^{\rm T}{\bf V}_2\right\|_F^2\right|  \nonumber\\
&\lesssim d_1\left\|\bar{\bf A}_1^{\rm T}{\bf V}_2\right\|_F^2\varepsilon_3. \label{eq-proof-thm-sub-bound1}
\end{align}

\emph{\bf Step 3)}
Now we will deem all basis vectors of ${\set X}_1$ separately as multiple one-dimensional subspaces, denoted by ${\set X}_{1,i}, 1\le i\le d_1$.
According to its definition and \eqref{eq-proof-thm-sub-orig-line-aff}, the affinity between ${\set X}_{1,i}$ and ${\set X}_2$ equals $\lambda_i$.
Actually we are interested in the relation between ${\set X}_{1,i}$ and ${\set X}_2$ after random projection.
This has been solved by Lemma $\ref{lem-line-sub}$, which means that, with probability at least $1 - P_1(\varepsilon_4,n)$,
\begin{equation}\label{eq-proof-thm-sub-bound-line}
\left|{\rm aff}_{{\set Y}_i}^2 - \overline{\rm aff}_{{\set Y}_i}^2\right| \lesssim \lambda_i^2(1-\lambda_i^2)\varepsilon_4,
\end{equation}
where 
\begin{align}
{\rm aff}_{{\set Y}_i}^2 &= \left\|\bar{\bf a}_{1, i}^{\rm T}{\bf V}_2\right\|^2, \label{eq-proof-thm-sub-aff-line}\\
\overline{\rm aff}_{{\set Y}_i}^2 &= \lambda_i^2+\frac{d_2}{n}(1-\lambda_i^2), \nonumber
\end{align}
denote, respectively, the affinity and its estimate between the compressed line ${\set Y}_{1,i}$ and the compressed subspace ${\set Y}_2$. 
Equation \eqref{eq-proof-thm-sub-aff-line} comes from that $\bar{\bf a}_{1,i}$ is the orthonormal basis for ${\set Y}_{1,i}$.

Considering the independence among these one-dimensional subspaces and using \eqref{eq-define-affinity-2}, \eqref{thm-sub-aff-change}, \eqref{eq-proof-thm-sub-bound-line}, and \eqref{eq-proof-thm-sub-aff-line}, we have, with probability at least $1 - d_1P_1(\varepsilon_4,n)$, 
\begin{align}
        \left|\|\bar{\bf A}_1^{\rm T}{\bf V}_2\|_F^2-\oaffYS\right| 
        &= \left|\sum_{i= 1}^{d_1} \left\|\bar{\bf a}_{1, i}^{\rm T}{\bf V}_2\right\|^2-\left(\affXS+\frac{d_2}{n}(d_1-\affXS)\right)\right|\nonumber\\ 
        &= \left|\sum_{i= 1}^{d_1} {\rm aff}_{{\set Y}_i}^2-\sum_{i= 1}^{d_1}\left(\lambda_i^2+\frac{d_2}{n}(1-\lambda_i^2)\right)\right| \nonumber\\
        &\le \sum_{i= 1}^{d_1} \left|{\rm aff}_{{\set Y}_i}^2-\overline{\rm aff}_{{\set Y}_i}^2\right| 
        \lesssim \sum_{i= 1}^{d_1} \lambda_i^2(1-\lambda_i^2)\varepsilon_4. \label{eq-proof-thm-sub-bound2}
\end{align}

\emph{\bf Step 4)}
Combining \eqref{eq-proof-thm-sub-bound1} and \eqref{eq-proof-thm-sub-bound2} by utilizing triangle inequality, we readily reach that, with probability at least $1 - \left(d_1(d_1\!-\!1)/2\right)P_3(\varepsilon_3,n) - d_1P_1(\varepsilon_4,n) $,
\begin{align*}
\left|\affYS-\oaffYS\right| 
\le& \left|\affYS-\left\|\bar{\bf A}_1^{\rm T}{\bf V}_2\right\|_F^2\right| + \left|\left\|\bar{\bf A}_1^{\rm T}{\bf V}_2\right\|_F^2-\oaffYS\right| \\
\lesssim & \, d_1\left\|\bar{\bf A}_1^{\rm T}{\bf V}_2\right\|_F^2\varepsilon_3 + \sum_{i= 1}^{d_1} \lambda_i^2(1-\lambda_i^2)\varepsilon_4.
\end{align*}
Recalling Lemma \ref{lemma-simp} and that $P_3(\varepsilon_3,n)$ decreases exponentially with respect to $n$, we have
\begin{align}
        \mathbb{P}\left(\left|\affYS-\oaff^2\right| >\sum_{i= 1}^{d_1} \lambda_i^2(1-\lambda_i^2)\varepsilon\right) \label{eq-proof-thm-sub-bound-tight}
        &\lesssim d_1P_1(\varepsilon_4,n) + \frac{d_1(d_1-1)}{2}P_3(\varepsilon_3,n) \\
        &\sim d_1P_1(\varepsilon,n) = \frac{4d_1}{\varepsilon^2n}, \nonumber
\end{align}
We then complete the proof of Lemma \ref{thm-sub-tight}. Finally, relaxing the bound in \eqref{eq-proof-thm-sub-bound-tight} by
$$
 \sum_{i= 1}^{d_1} \lambda_i^2(1-\lambda_i^2) \le \sum_{i= 1}^{d_1} \lambda_i^2 = \affXS,
$$
Theorem \ref{thm-sub} is proved.

\subsection{Proof of Corollary \ref{cor-sub-distance}}
\label{Appendixproof-cor-sub-distance}

By reshaping \eqref{sub-distance-estimate}, we have
\begin{equation}\label{eq-proof-cor-sub-distance-1}
\dYS \!-\! \odYS = \left(\dYS \!-\! \dXS \right) + \frac{d_2}{n}\left(\dXS-\frac{d_2\!-\!d_1}{2}\right).
\end{equation}
Using Lemma \ref{lem-aff-to-dist} in \eqref{eq-proof-cor-sub-distance-1}, we are ready to verify
\begin{equation}\label{eq-proof-cor-sub-distance-2}
\dYS - \odYS = \oaffYS - \affYS.
\end{equation}
Now let check the bound in \eqref{thm-sub-aff-bound-tight},
\begin{equation}
	\sum_{i= 1}^{d_1} \lambda_i^2(1\!-\!\lambda_i^2) \le \sum_{i= 1}^{d_1} (1\!-\!\lambda_i^2) 
        = d_1 \!-\! \affXS 
        \le \dXS. \label{eq-proof-cor-sub-distance-3}
\end{equation}
Combing \eqref{eq-proof-cor-sub-distance-2} and \eqref{eq-proof-cor-sub-distance-3} in Lemma \ref{thm-sub-tight}, the proof is complete.

\subsection{Proof of Theorem \ref{thm-rip}}
\label{Appendixproof-thm-rip}

Jointly applying the fact of $\dXS \ge (d_2-d_1)/2$ by Lemma \ref{lem-aff-to-dist} and triangle inequality in \eqref{eq-proof-cor-sub-distance-1}, we have
\begin{equation}\label{eq-proof-thm-rip-2}
\left|\dYS \!-\! \dXS\right| \!\le\! \left|\dYS \!-\! \odYS\right| \!+\!  \frac{d_2}{n}\!\left(\!\dXS\!-\!\frac{d_2\!-\!d_1}{2}\right).
\end{equation}
Inserting \eqref{eq-proof-thm-rip-2} in \eqref{sub-dist-bound}, we have
\begin{equation}\label{eq-proof-thm-rip-3}
\mathbb{P}\!\left(\!\left|\dYS\!-\!\dXS\right|\!>\!\dXS\varepsilon_1\!+\!\frac{d_2}{n}\!\left(\!\dXS\!-\!\frac{d_2\!-\!d_1}{2}\right)\!\right)  \!\lesssim \frac{4d_1}{\varepsilon_1^2n}.
\end{equation}
By redefining $\varepsilon$ as
\begin{equation}\label{eq-proof-thm-rip-4}
	\dXS\varepsilon = \dXS\varepsilon_1+\frac{d_2}{n}\left(\dXS-\frac{d_2-d_1}{2}\right),
\end{equation}
we have
\begin{equation}\label{eq-proof-thm-rip-5}
	\frac{4d_1}{\varepsilon_1^2n} = \frac{4d_1}{\left(\varepsilon - \frac{d_2}{n} + \frac{d_2(d_2-d_1)}{2n\dXS}\right)^2n} <
	\frac{4d_1}{\left(\varepsilon - \frac{d_2}{n}\right)^2n}.
\end{equation}
Using \eqref{eq-proof-thm-rip-4} and \eqref{eq-proof-thm-rip-5} in \eqref{eq-proof-thm-rip-3}, the proof is complete.

\section{Numerical verification}

In this section, the main result of Theorem \ref{thm-sub} is evaluated by numerical simulations. 
In order to save computation, we randomly generate two subspaces in the following steps. 
\begin{enumerate}
\item
Given $d_1\le d_2 \ll N$, generate an orthonormal matrix ${\bf W} = \left[{\bf w}_1,\cdots,{\bf w}_{d_1+d_2}\right] \in\mathbb{R}^{N\times (d_1+d_2)}$.
\item
Let ${\bf U}_2 = \left[{\bf w}_1,\cdots,{\bf w}_{d_2}\right]$ be the orthonormal basis for subspace ${\set X}_2$.
\item
Given affinity $\affX$, randomly choose $\hat\lambda_i, 1\le i \le d_1$ from the uniform distribution in $[0,1]$ and then scale them to the affinity, i.e.
$$
	\lambda_i = \affX\cdot\frac{\hat\lambda_i}{\left(\sum_{i=1}^{d_1}\hat\lambda_i^2\right)^{\frac12}}.
$$
\item
Calculate the orthonormal basis for subspace ${\set X}_1$ as
\begin{align*}
	{\bf U}_1 = \big[\lambda_1{\bf w}_1+{(1-\lambda_1^2)^{\frac12}}{\bf w}_{d_2+1},\cdots,
 \lambda_{d_1}{\bf w}_{d_1}+{(1-\lambda_{d_1}^2)^{\frac12}}{\bf w}_{d_2+d_1}\big].
\end{align*}
\end{enumerate}
With this method, we can generate two subspaces with any given affinity, which are ready for projection.
In addition, several subspace clustering algorithms are conducted on both compressed synthetic data and real-world data to verify the motivation and application of this work.

\subsection{Affinities before and after random projection}

In the first experiment, the estimate of the compressed affinity \eqref{thm-sub-aff-change} is verified in the condition of $(N,n)=(500, 200)$ and $(d_1,d_2) = (5,10)$. 
The original affinity in the ambient space is chosen as $\affXS=1,2,3,4$, respectively. 
For each $\affXS$, a random Gaussian matrix is generated and used to project the two subspaces into the compressed space, where the compressed affinity $\affYS$ is calculated. 
The frequencies of the compressed affinities obtained from $1{\rm E}5$ trials as well as with their theoretical estimates are demonstrated in Fig. \ref{simfigure1}. 
One may read that the proposed estimate is rather accurate and the compressed affinities concentrate on their theoretical estimates.

\begin{figure}[!t]
\begin{center}
\includegraphics[width=0.7\textwidth]{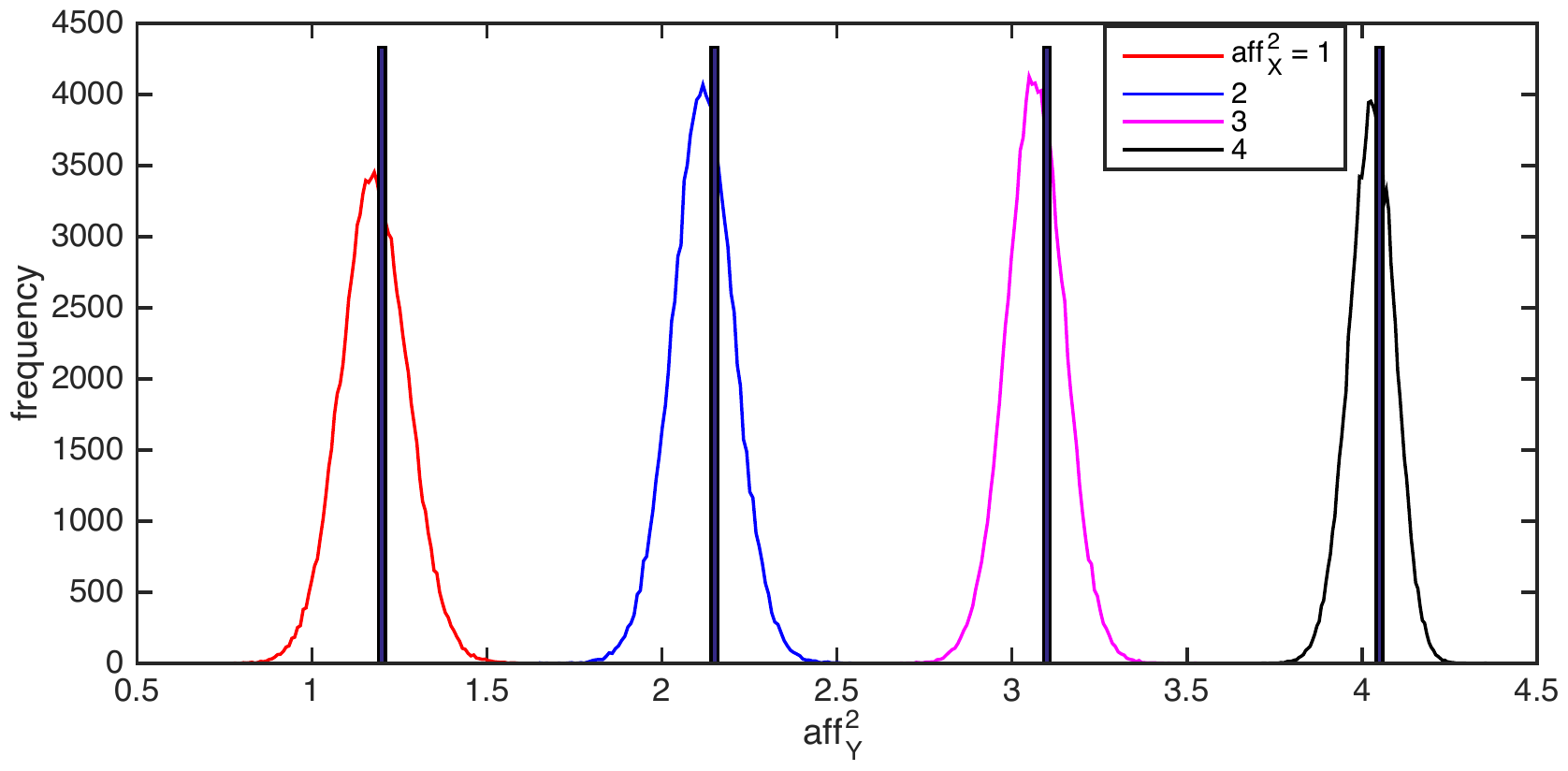}
\caption{This figure demonstrates the experimental frequency (denoted by curves) and the theoretical estimate (denoted by bars) of the compressed affinity, where $(N,n)=(500,200)$, $(d_1,d_2)=(5,10)$, and the original affinities are fixed as $1,2,3$, and $4$. 
The frequencies are calculated by 1E5 trials.}\label{simfigure1}
\end{center}
\end{figure}

In the second experiment, the estimate of the compressed affinity is further tested for all possible original affinities and by various subspace dimension combinations, where the dimensions of the ambient space and compressed space are the same as those in the first experiment. 
Here $(d_1,d_2)$ is chosen from a candidate set and the original affinity varies from $0$ to its maximum, i.e., $d_1\le d_2$. 
For each case, two original subspaces and a random Gaussian matrix are generated, then the compressed affinity is calculated after projection. 
After repeating 500 times, the frequencies at different compressed affinities are computed and normalized by its maximum, i.e., the compressed affinity with the highest appearance is assigned 1 and the others are smaller than 1. 
Then the normalized frequencies for all cases are plotted in Fig. \ref{simfigure2}, where the blue line denotes the theoretical estimate. 
This result further verifies that the compressed affinities of various dimensions of subspaces display the concentration property, as shows in Theorem \ref{thm-sub}.
 
\begin{figure}[!t]
\begin{center}
\includegraphics[width=0.7\textwidth]{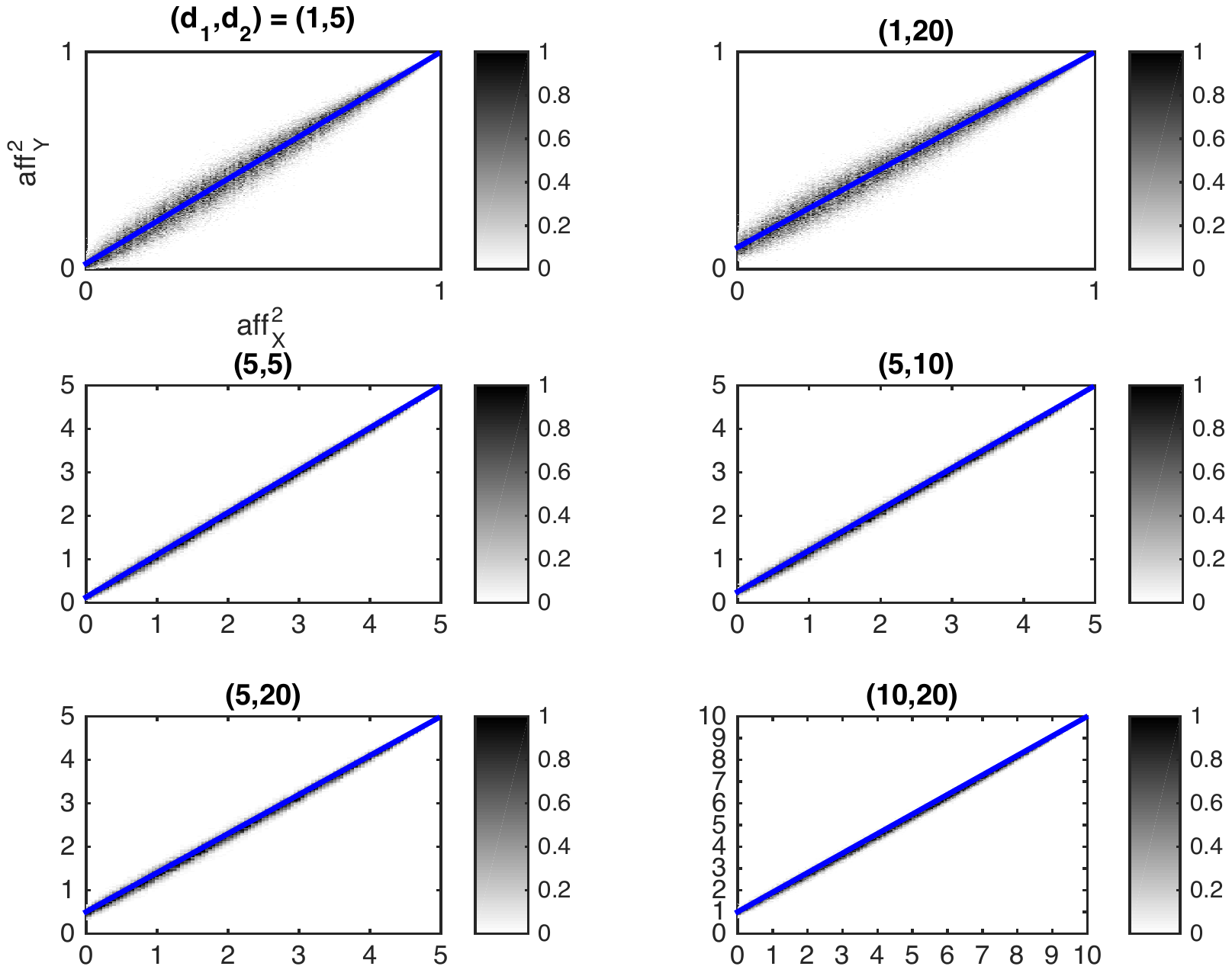}
\caption{This figure demonstrates the experimental compressed affinity (which frequency is denoted by gray area) and the theoretical estimate (denoted by blue line), where $(N,n)=(500,200)$ and $(d_1,d_2)$ are displayed on the title.}\label{simfigure2}
\end{center}
\end{figure}

The third experiment tests the effect of $N$ and $n$ in Theorem \ref{thm-sub}. 
By fixing $(d_1,d_2)=(5,10)$, the compressed affinity of two subspaces being projected from an $N$-dimension space to an $n$-dimension space, where $(N, n)$ is chosen from a candidate set, is shown.
The result is plotted in Fig. \ref{simfigure3} by using the same way as that in the second experiment. 
One may readily find that by increasing $n$, the compressed affinity demonstrates better concentration. 
Whereas the dimension of the original space, $N$, has no effect on the  concentration behavior.
The observation agrees with Theorem \ref{thm-sub}.

\begin{figure}[!t]
\begin{center}
\includegraphics[width=0.7\textwidth]{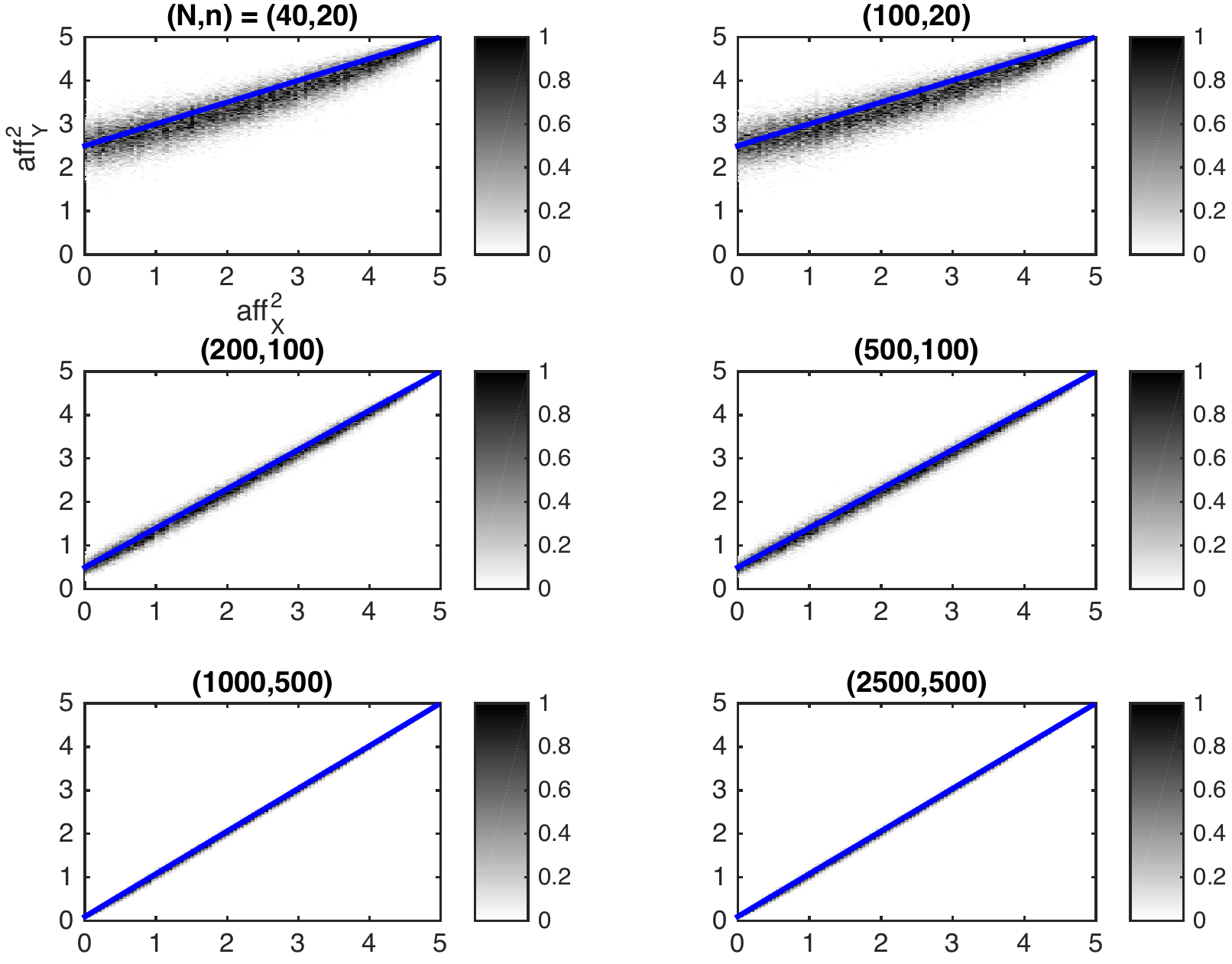}
\caption{This figure demonstrates the experimental compressed affinity (which frequency is denoted by gray area) and the theoretical estimate (denoted by blue line), where $(d_1,d_2)=(5, 10)$ and $(N,n)$ are displayed on the title. }\label{simfigure3}
\end{center}
\end{figure}

\subsection{Clustering performance of CSC}

In the forth experiment, we conduct CSC on synthetic data using several clustering algorithms.
Two low-dimensional subspaces of dimensions $d_1=d_2=7$ with $3$ intersect dimensions are selected in a $200$-dimensional ambient space. In each subspace $60$ samples are randomly generated.
The sampled points are then compressed by Gaussian random matrix to an $n$-dimensional ambient space with $n=200, 150, \ldots, 5$, respectively.
After compression, the $120$ $n$-dimensional data points are clustered by several algorithms including SSC \cite{Elhamifar2009Sparse}, square-root SSC (SR-SSC) \cite{SR-SSC}, and low-rank representation (LRR) \cite{LRR}.
The clustering error rate, which is averaged from $30$ independent compression tests, of the three algorithms are plotted in the top figure of Fig. \ref{simfigure5},
where the affinity after compression are plotted in the bottom figure.
One may read that when the compression ratio $n/N$ decreases, along with the increase of the compressed affinity, the clustering error increases synchronously.
This verifies our theory and the motivation of our work, i.e., data compression before subspace clustering can be adopted to save computation, but too much compression can cause high error rate.

\begin{figure}[!t]
\begin{center}
\includegraphics[width=0.7\textwidth]{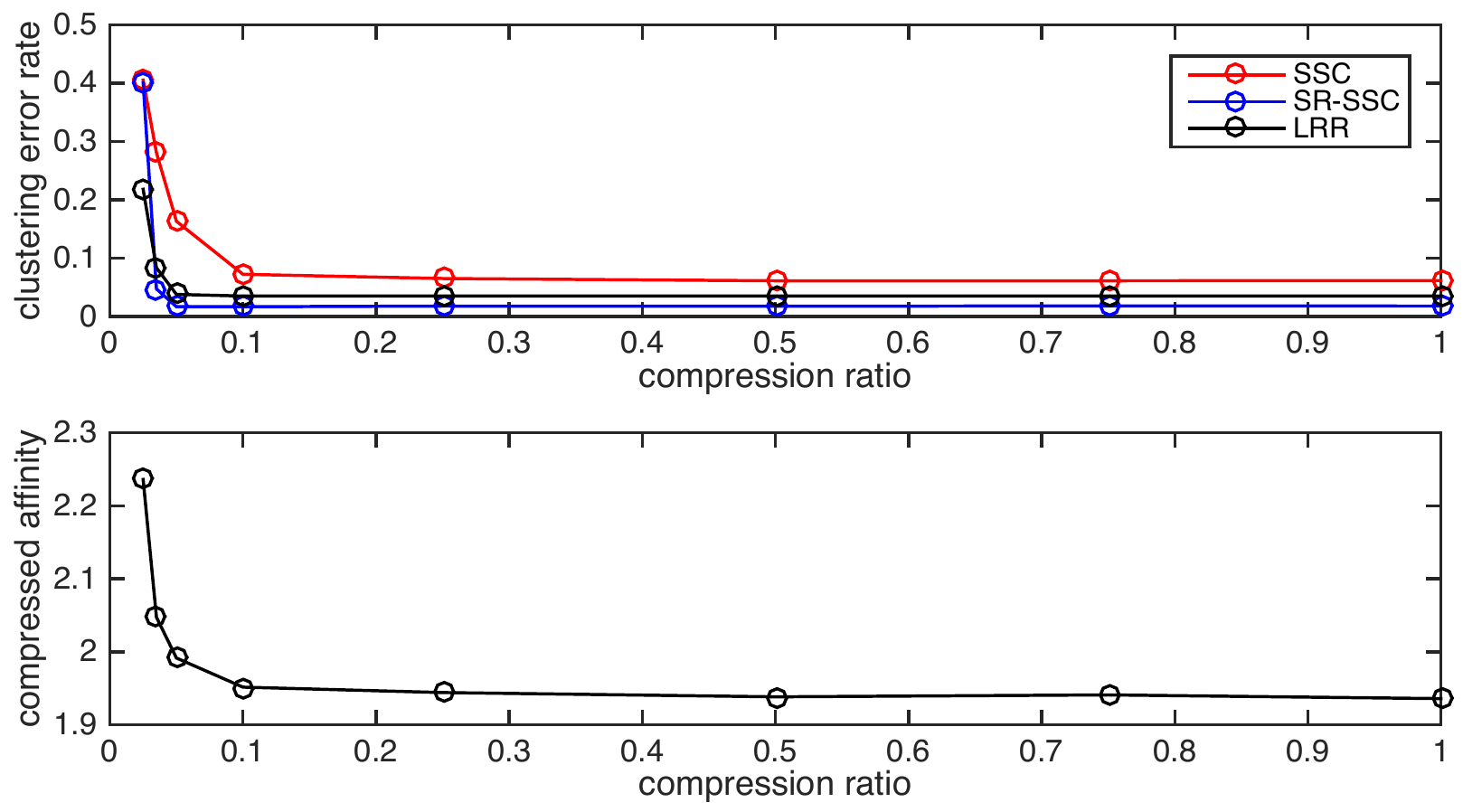}
\caption{Clustering error (top) and compressed affinity (bottom) v.s. compression ratio $(n/N)$ on synthetic data. \label{simfigure5}}
\end{center}
\end{figure}

\begin{figure}[!t]
  \centering
  \subfigure[YaleB]{
    \label{fig:realworlddata:a} %% label for first subfigure
	\includegraphics[width=0.25\textwidth]{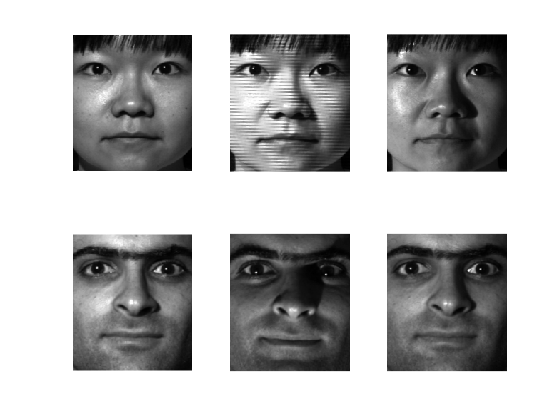}}
  \subfigure[MNIST]{
    \label{fig:realworlddata:b} %% label for second subfigure
\includegraphics[width=0.25\textwidth]{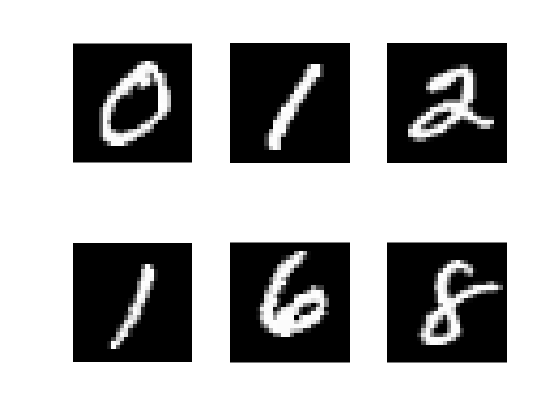}}
\caption{Examples of the two real-world databases.\label{fig:examples}}
\end{figure}

In the last experiment, we try CSC on two real-world databases, including YaleB \cite{EYaleB} and MNIST handwritten digits images \cite{MNIST} as exemplified in Fig. \ref{fig:examples}, by using the above three clustering algorithms.
\begin{itemize}
\item
For the face images, we choose the $1$st and the $5$th subjects, where each subject has $64$ face images of dimension $N=32256$. 
Each image is then compressed to $n=100, 75, \ldots, 5$ dimensions, respectively, 
by Gaussian random projection.
The clustering error averaged from $50$ independent tests and the compressed affinity are plotted in Fig. \ref{simfigure7}.
\item
For the handwritten images, we choose digits ``$1$'' and ``$2$'', where each digit has $300$ handwritten images of dimension $N=784$. 
Each image is then compressed to $n=500, 400, \ldots, 10$. 
The clustering error averaged from $20$ independent tests and the compressed affinity are plotted in Fig. \ref{simfigure8}.
\end{itemize}
The trend of curves in Fig. \ref{simfigure7} and Fig. \ref{simfigure8} are exactly the same as the ones in Fig. \ref{simfigure5}.
Therefore, the real-world data also verify the motivation and future potential applications of this work.
When the dimension of new ambient space is large enough, one may do random compression to reduce the data size while keeping the distance between latent subspaces, 
which plays an important role in solving subspace clustering problems.

\begin{figure}[!t]
\begin{center}
\includegraphics[width=0.7\textwidth]{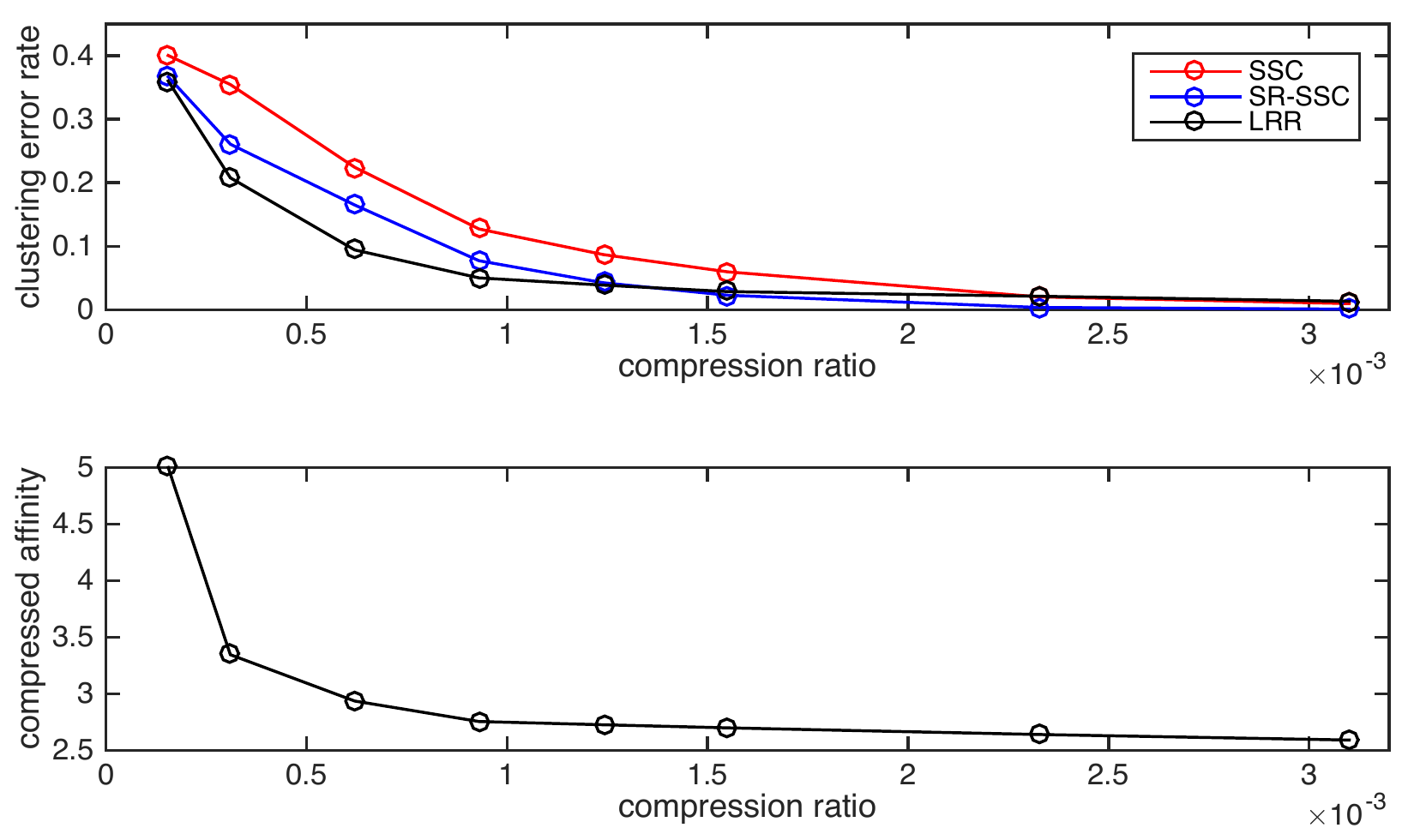}
\caption{Clustering error and compressed affinity v.s. compression ratio $(n/N)$ on YaleB face images.\label{simfigure7}}
\end{center}
\end{figure}

\begin{figure}[!t]
\begin{center}
\includegraphics[width=0.7\textwidth]{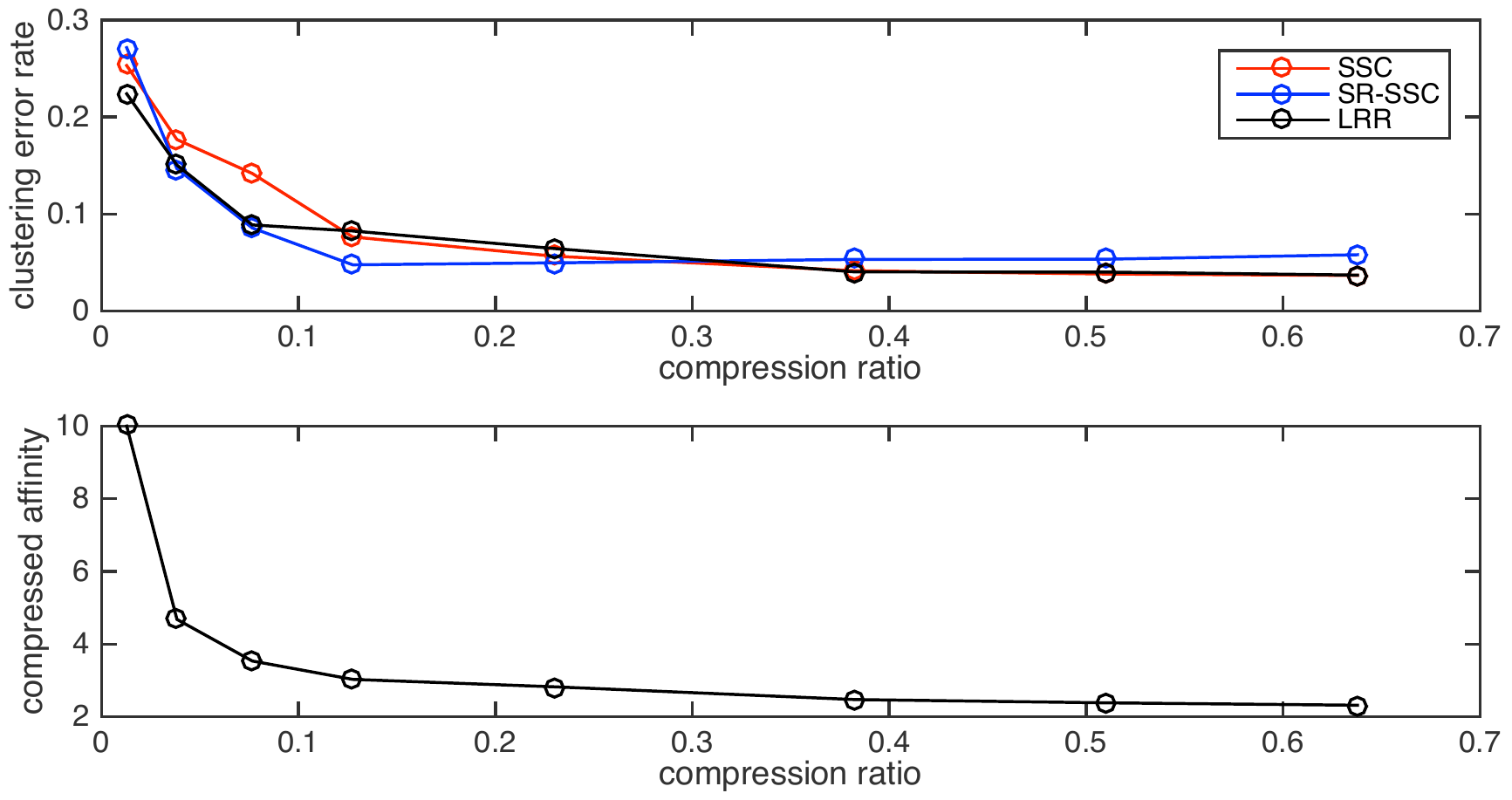}
\caption{Clustering error and compressed affinity v.s. compression ratio $(n/N)$ on MNIST handwritten digits images.\label{simfigure8}}
\end{center}
\end{figure}

\section{Conclusion}

This work generalizes the JL Lemma and the RIP from finite signal set and sparse signals, respectively, to subspaces. 
We study the distance-preserving property of Gaussian random projection for subspaces.
By introducing the projection $F$-norm distance and build a metric space, the connection between affinity and distance are revealed for the first time.
We then theoretically prove that with high probability the affinity or distance between two projected subspaces are concentrated on their estimates. When the new ambient dimension is sufficiently large, the affinity and distance between two subspaces almost remain unchanged after random projection.
Finally, the main contribution is generalized to a finite set of subspaces. In addition, we provided numerical simulations for validation. 

\section{Appendix}

\subsection{Proof of Lemma \ref{lem-aff-to-dist}}
\label{Appendixproof-lem-aff-to-dist}

Denote the orthonormal basis of subspace ${\set X}_k$ by ${\bf U}_k=\left[{\bf u}_{k,1}, ..., {\bf u}_{k, d_k}\right], k = 1, 2$. According to the definition of distance in \eqref{eq-define-distance},  we have
\begin{align}
        \quad D^2(\set{X}_1, \set{X}_2)
        &= \frac{1}{2}\|{\bf P}_{1}-{\bf P}_{2}\|_F^2 \nonumber\\
        &= \frac{1}{2}\left\|{\bf U}_1{\bf U}_1^{\rm T}-{\bf U}_2{\bf U}_2^{\rm T}\right\|_F^2 \nonumber\\
        &= \frac{1}{2}\bigg\|\sum_{i=1}^{d_1} {\bf u}_{1,i}{\bf u}_{1,i}^{\rm T} - \sum_{j=1}^{d_2} {\bf u}_{2,j}{\bf u}_{2,j}^{\rm T}\bigg\|_F^2.\label{eq-proof-lem-aff-to-dist-1}
\end{align}
Denote the $l$th entry of ${\bf u}_{k,i}$ by $u_{k, li}$ and expand the RHS of \eqref{eq-proof-lem-aff-to-dist-1}, we have
\begin{align}
        \quad D^2(\set{X}_1, \set{X}_2)
        &= \!\frac{1}{2}\sum_{l=1}^{n} \sum_{m=1}^{n} \!\!\bigg(\!\sum_{i=1}^{d_1} u_{1,li}u_{1,mi} - \!\sum_{j=1}^{d_2} u_{2,lj}u_{2,mj}\!\bigg)^2 \nonumber\\
        &=: \frac{1}{2}\sum_{l=1}^{n} \sum_{m=1}^{n} \left(A_1+A_2-2B\right),\label{eq-proof-lem-aff-to-dist-sum}
\end{align}
where
\begin{align}
A_k &= \bigg(\sum_{i=1}^{d_k} u_{k,li}u_{k,mi}\bigg)^2\nonumber\\
&= \sum_{i=1}^{d_k} u_{k,li}^2u_{k,mi}^2 + \!\!\sum_{1\le i\ne j\le d_k} \!\!\!\!u_{k,li}u_{k,mi}u_{k,lj}u_{k,mj},\label{eq-proof-lem-aff-to-dist-Ak}\\
B &= \sum_{i=1}^{d_1} u_{1,li}u_{1,mi} \sum_{j=1}^{d_2} u_{2,lj}u_{2,mj}. \label{eq-proof-lem-aff-to-dist-B}
\end{align}
We will study the three items separately. First, using \eqref{eq-proof-lem-aff-to-dist-Ak} in \eqref{eq-proof-lem-aff-to-dist-sum} and changing the order of summation, we have
\begin{align*}
	\sum_{l=1}^{n} \sum_{m=1}^{n}A_k 
	=&\sum_{i=1}^{d_k} \bigg(\sum_{l=1}^{n} u_{k,li}^2\sum_{m=1}^{n} u_{k,mi}^2\bigg) 
	+ \!\!\sum_{1\le i\ne j\le d_k} \!\!\!\bigg(\sum_{l=1}^{n} u_{k,li}u_{k,lj}\sum_{m=1}^{n} u_{k,mi}u_{k,mj}\bigg)\\
	=& \sum_{i=1}^{d_k} \|{\bf u}_{k, i}\|^2\|{\bf u}_{k, i}\|^2 + \!\!\sum_{1\le i\ne j\le d_k} \!\!\!{\bf u}_{k, i}^{\rm T}{\bf u}_{k, j}{\bf u}_{k, i}^{\rm T}{\bf u}_{k, j}.
\end{align*}
Considering that ${\bf u}_{k,i}, 1\le i\le d_k$ are columns drawn from an orthonormal matrix, we have
\begin{equation}\label{eq-proof-lem-aff-to-dist-Ak-simp} 
	\sum_{l=1}^{n} \sum_{m=1}^{n}A_k = \sum_{i=1}^{d_k} 1 + \!\!\sum_{1\le i\ne j\le d_k} \!\!\!\!0 
	= d_k.
\end{equation}
Now we check the last item in \eqref{eq-proof-lem-aff-to-dist-sum}. Using \eqref{eq-proof-lem-aff-to-dist-B} in \eqref{eq-proof-lem-aff-to-dist-sum} and changing the order of summation, we have
\begin{align*}
	\sum_{l=1}^{n} \sum_{m=1}^{n}B &= \sum_{i=1}^{d_1}\sum_{j=1}^{d_2} \left(\sum_{l=1}^{n} u_{1,li}u_{2,lj}\sum_{m=1}^{n} u_{1,mi}u_{2,mj}\right)\\
	&= \sum_{i=1}^{d_1}\sum_{j=1}^{d_2} \left|{\bf u}_{1, i}^{\rm T}{\bf u}_{2, j}\right|^2.
\end{align*}
Recalling the definition of affinity in \eqref{eq-define-affinity}, we have
\begin{equation}\label{eq-proof-lem-aff-to-dist-B-simp} 
	\sum_{l=1}^{n} \sum_{m=1}^{n}B = \aff ^2(\set{X}_1, \set{X}_2).
\end{equation}
We then complete the proof by inserting \eqref{eq-proof-lem-aff-to-dist-Ak-simp} and \eqref{eq-proof-lem-aff-to-dist-B-simp} in \eqref{eq-proof-lem-aff-to-dist-sum}.

\subsection{Proof of Lemma \ref{lemma-simp}}
\label{proof-lemma-simp}

By introducing $\varepsilon = \varepsilon_1 + c_2\varepsilon_2/c_1$, we have $\varepsilon_1 = (1 - {1}/{m})\varepsilon$, and $\varepsilon_2 = {c_1}\varepsilon/(mc_2)$ for all $m \in \mathbb{N}^+$. Using \eqref{Appendixequation1.5} in \eqref{Appendixequation1}, we have, for large $n$,
\begin{align*}
\mathbb{P}\left(X(n) > c_1\varepsilon \big| n\right) &\le p\left(\left(1 - 1/m\right)\varepsilon, n\right) + q\left(\frac{c_1}{mc_2}\varepsilon, n\right) \\
&\sim p\left(\varepsilon\left(1 - {1}/{m}\right), n\right).
\end{align*}
Let $m$ approach infinity, and then \eqref{Appendixequation2} is obtained.

\subsection{Proof of Lemma \ref{L4}}
\label{proof-L4}

First, we introduce Lemma \ref{L1}, from which Lemma \ref{L4} is extended.

\begin{lem}\label{L1}
Let ${\bf a}_1, {\bf a}_2 \in\mathbb{R}^n$ are two independent Gaussian random vectors. Since $\|{\bf a}_1\|^2 \sim \chi_n^2/n$ and $\|{\bf a}_2\|^2 \sim \chi_n^2/n$ are independent, 
$
F_{n,n}:={\|{\bf a}_1\|^2}/{\|{\bf a}_2\|^2}
$
follows an F-distribution with parameters $(n, n)$, so we have
\begin{equation}\label{eq-p1epsilon}
        \mathbb{P}(|F_{n,n}-1|>\varepsilon) 
        \lesssim \frac{4}{\varepsilon^2n} = P_1(\varepsilon,n),
\end{equation}
when $n$ is large enough.
\end{lem}

\begin{proof}
According to the properties of F-distribution, we have
$        
	\mathbb{E}F_{n,n} = \frac{n}{n-2},
      {\rm Var}(F_{n,n}) = \frac{4n(n-1)}{(n-2)^2(n-4)}.
$
Then \eqref{eq-p1epsilon} is verified by using Chebyshev's inequality as
\begin{align*}
        \mathbb{P}(|F_{n,n}-1|>\varepsilon) 
        &\le \frac{\mathbb{E}|F_{n,n}-1|^2}{\varepsilon^2} \\
        &= \frac{{\rm Var}(F_{n,n})+(\mathbb{E}F_{n,n}-1)^2}{\varepsilon^2} \\
        &= \frac{4(n+2)}{\varepsilon^2(n-2)(n-4)}
        \sim \frac{4}{\varepsilon^2n}.
\end{align*}
\end{proof}

Now we begin the proof of Lemma \ref{L4} by introducing an assistant vector 
\begin{equation}\label{eq-L4-proof-a3}
{\bf w} = \frac{{\bf p}-\omega{\bf q}}{\sqrt{1-\omega^2}},
\end{equation}
which is orthogonal to ${\bf q}$. This can be verified by
\begin{align}
\mathbb{E}{\bf w}{\bf q}^{\rm T} &= \frac{1}{\sqrt{1-\omega^2}}\left(\mathbb{E}{\bf p}{\bf q}^{\rm T}-\omega\mathbb{E}{\bf q}{\bf q}^{\rm T}\right) \nonumber\\
&= \frac{1}{\sqrt{1-\omega^2}}\left(\frac{\omega}{n}{\bf I}_n-\frac{\omega}{n}{\bf I}_n\right) 
= {\bf 0}.\label{eq-L4-proof-Ea3Ta2}
\end{align}
Using the above orthogonality and the given condition, we further write
\begin{align*}
\mathbb{E}{\bf w}{\bf w}^{\rm T} &= \frac{1}{\sqrt{1-\omega^2}}\mathbb{E}{\bf w}{\bf p}^{\rm T} 
= \frac{1}{1-\omega^2}\left(\mathbb{E}{\bf p}{\bf p}^{\rm T}-\omega\mathbb{E}{\bf q}{\bf p}^{\rm T}\right) \\
&= \frac{1}{1-\omega^2}\left(\frac{1}{n}{\bf I}_n-\frac{\omega^2}{n}{\bf I}_n\right) 
= \frac{1}{n}{\bf I}_n.
\end{align*}
These show that ${\bf q}$ and ${\bf w}$ are independent Gaussian random vectors. Following Lemma \ref{L1}, we denote ${\|{\bf w}\|^2}/{\|{\bf q}\|^2}$ by $F_{n,n}$. 

By representing ${\bf p}$ with ${\bf w}$ using \eqref{eq-L4-proof-a3}, consequently, we have
\begin{align}
\frac{\|{\bf p}\|^2}{\|{\bf q}\|^2} &= \frac{\|\omega{\bf q}+\sqrt{1-\omega^2}{\bf w}\|^2}{\|{\bf q}\|^2}\nonumber \\
&= \omega^2+(1-\omega^2)\frac{\|{\bf w}\|^2}{\|{\bf q}\|^2}+2\omega\sqrt{1-\omega^2}\frac{\|{\bf w}\|}{\|{\bf q}\|}\cos\theta\nonumber\\
&= \omega^2+(1\!-\!\omega^2)F_{n,n}+2\omega\sqrt{1\!-\!\omega^2}\sqrt{F_{n,n}}\cos\theta, \label{eq-L4-proof-a1toa3}
\end{align}
where $\theta$ is the angle between ${\bf q}$ and ${\bf w}$.  
Recalling Lemma \ref{L1} and Lemma \ref{L3}, we have, with probability at least $1-P_1(\varepsilon_5,n) - P_3(\varepsilon_6,n)$,
\begin{align}
&\left|(1-\omega^2)(F_{n,n}-1)+2\omega\sqrt{1-\omega^2}\sqrt{F_{n,n}}\cos\theta\right| \nonumber \\
\le& (1-\omega^2)\varepsilon_5+2\omega\sqrt{1-\omega^2}\sqrt{1+\varepsilon_5}\varepsilon_6, \label{eq-L4-proof-asum}\\
\sim& (1-\omega^2)\varepsilon_5+2\omega \sqrt{1-\omega^2}\varepsilon_6, \label{eq-L4-proof-asum1}
\end{align}
where $\sqrt{1+\varepsilon_5}$ in \eqref{eq-L4-proof-asum} comes from \eqref{eq-p1epsilon} and it is then approximated by one in \eqref{eq-L4-proof-asum1} because $\varepsilon_5$ is a small quantity.

Now we will use Lemma \ref{lemma-simp} to simplify \eqref{eq-L4-proof-asum1}. Because $P_3(\varepsilon_6,n)$ decreases exponentially with respect to $n$ and the condition of Lemma \ref{lemma-simp} is satisfied, we have, when $n$ is large,
\begin{align}
        &\mathbb{P}\left(\left|(1\!-\!\omega^2)(F_{n,n}\!-\!1)\!+\!2\omega\sqrt{1\!-\!\omega^2}\sqrt{F_{n,n}}\cos\!\theta\right|\!>\!(1\!-\!\omega^2)\varepsilon\right) \nonumber\\
        &\lesssim  P_1(\varepsilon_5,n) + P_3(\varepsilon_6,n) 
        \sim P_1(\varepsilon,n) = \frac{4}{\varepsilon^2n}.\label{eq-L4-proof-epsilon12}
\end{align}

Combining \eqref{eq-L4-proof-a1toa3} and \eqref{eq-L4-proof-epsilon12}, the proof is completed.

\subsection{Proof of Lemma \ref{C1}}
\label{proof-C1}

Lemma \ref{C1} is a corollary of the following Lemma.

\begin{lem}\label{L2}
Let ${\bf w}={\bf a}/{\|{\bf a}\|}$, where ${\bf a} \in\mathbb{R}^n$ is a Gaussian random vector. For any support ${\set T}\subset [1:n]$ with cardinality $d := |{\set T}|$, we have
\begin{equation}\label{eq-p2epsilon}
    \mathbb{P}\left(\left|\|{\bf w}_{\set T}\|^2 - \frac{d}{n}\right|>\varepsilon\right) 
    < \frac{2d}{\varepsilon^2n^2} = P_2(\varepsilon,n),
\end{equation}
where ${\bf w}_{\set T}$ is composed by the entries of $\bf w$ supported on $\set T$.
\end{lem}

\begin{proof}
By calculating, we have
$
        \mathbb{E}\|{\bf w}_{\set T}\|^2 = \frac{d}{n},
        {\rm Var}\left(\|{\bf w}_{\set T}\|^2\right) = \frac{2d(n-d)}{n^2(n+2)}.
$
Then \eqref{eq-p2epsilon} is verified by using Chebyshev's inequality as
\begin{align*}
    \mathbb{P}\left(\left|\|{\bf w}_{\set T}\|^2 - \frac{d}{n}\right|>\varepsilon\right)
    \le \frac{2d(n-d)}{\varepsilon^2n^2(n+2)}
    < \frac{2d}{\varepsilon^2n^2}.
\end{align*}
\end{proof}

Recalling the definition of $\bf w$ in Lemma \ref{L2}, we may let $w_i = \cos\phi_i$, where $\phi_i$ denotes the angle between $\bf w$ and the $i$th coordinate axis, ${\bf e}_i$. Because the relation of $\bf w$ with respect to ${\bf E} = [{\bf e}_1,\cdots,{\bf e}_d]$ and that of ${\bf u}$ with respect to ${\bf V}$ are exactly identical, \eqref{eq-p2epsilon1} can be readily verified.

\subsection{Proof of Lemma \ref{lem-math}}
\label{proof-lem-math}

Suppose that $\mathcal S$ is a $d$-dimensional subspace in $\mathbb{R}^n$.
The columns of ${\bf A}=\left[{\bf a}_1,{\bf a}_2,\dots,{\bf a}_d\right]\in\mathbb{R}^{n\times d}$
constitute a basis of $\mathcal S$.
We normalize the columns of ${\bf A}$ and obtain a normal basis of $\mathcal S$ as the following
\begin{equation}\label{eq_L1_0}
	\bar{\bf A} = \left[\bar{\bf a}_1,\bar{\bf a}_2,\dots,\bar{\bf a}_d\right] = \left[\frac{{\bf a}_1}{\|{\bf a}_1\|},\frac{{\bf a}_2}{\|{\bf a}_2\|},\dots,\frac{{\bf a}_d}{\|{\bf a}_d\|}\right].
\end{equation}
Furthermore, we can apply Gram–Schmidt process on $\bar{\bf A}$ to obtain an orthonormal basis of $\mathcal S$ as the following
\begin{equation}\label{eq_L1_1}
{\bf v}_i = \frac{\tilde{\bf v}_i}{\|\tilde{\bf v}_i\|}, \quad i = 1,2,\dots,d,
\end{equation}
where
\begin{equation}\label{eq_L1_2}
\tilde{\bf v}_i =\bar{\bf a}_i-\sum_{m=1}^{i-1} \left(\bar{\bf a}_i^{\rm T}{\bf v}_m\right){\bf v}_m, \quad i = 1,2,\dots,d.
\end{equation}
Notice that the index $i$ in \eqref{eq_L1_2} should start from $1$
and increase to $d$.
We denote the matrix $\left[{\bf v}_1,{\bf v}_2,\dots,{\bf v}_d\right]$ as $\bf V$.

We have ${\bf V}=\bar{\bf A}\bar{\bf G}$, where $\bar{\bf G}$ is an upper triangular matrix. Accordingly, $\bar{\bf U}=\bar{\bf G}-{\bf I}$ is also upper triangular and 
\begin{equation}\label{eq_L1_3}
{\bf v}_i = \bar{\bf a}_i+\sum_{j=1}^i \bar u_{ji}\bar{\bf a}_j.
\end{equation}
Using \eqref{eq_L1_3} and \eqref{eq_L1_2} in \eqref{eq_L1_1}, we have
\begin{equation}\label{eq_L1_4}
        {\bf v}_i = \frac{1}{\|\tilde{\bf v}_i\|} \bigg(\bar{\bf a}_i-\sum_{m=1}^{i-1} \bar{\bf a}_i^{\rm T}{\bf v}_m\bigg(\bar{\bf a}_m+\sum_{j=1}^m \bar u_{jm}\bar{\bf a}_j\bigg)\bigg).
\end{equation}
By switching the order of the summations, \eqref{eq_L1_4} can be reformulated as
\begin{align}
        {\bf v}_i &= \frac{1}{\|\tilde{\bf v}_i\|} \bigg(\bar{\bf a}_i 
        - \sum_{j=1}^{i-1}\bigg(\bar{\bf a}_i^{\rm T}{\bf v}_j+\sum_{m=j}^{i-1}\left(\bar{\bf a}_i^{\rm T}{\bf v}_m\right)\bar u_{jm}\bigg)
        \bar{\bf a}_j\bigg)\nonumber\\
        &= \frac{\bar{\bf a}_i}{\|\tilde{\bf v}_i\|} 
        - \sum_{j=1}^{i-1}\frac{\bar{\bf a}_i^{\rm T}{\bf v}_j+\sum_{m=j}^{i-1}\left(\bar{\bf a}_i^{\rm T}{\bf v}_m\right)\bar u_{jm}}{\|\tilde{\bf v}_i\|}\bar{\bf a}_j. \label{eq_L1_5}
\end{align}
Comparing \eqref{eq_L1_3} and \eqref{eq_L1_5}, we readily get
\begin{align}
	\bar u_{ii} &= \frac{1}{\|\tilde{\bf v}_i\|} - 1,\quad \forall i, \label{eq_L1_6}\\
	\bar u_{ji} &= - \frac{1}{\|\tilde{\bf v}_i\|}\bigg(\bar{\bf a}_i^{\rm T}{\bf v}_j+\sum_{m=j}^{i-1}\left(\bar{\bf a}_i^{\rm T}{\bf v}_m\right)\bar u_{jm}\bigg), \,\,\, \forall j < i. \label{eq_L1_7}
\end{align}
We will first study \eqref{eq_L1_6} and then turn to \eqref{eq_L1_7}. 
Plugging \eqref{eq_L1_2} into \eqref{eq_L1_6} and noticing that both $\bar{\bf a}_i$ and ${\bf v}_m$ have been normalized, we have 
\begin{align}\label{eq_L1_9}
	\bar u_{ii} &= \frac{1}{\|\bar{\bf a}_i-\sum_{m=1}^{i-1} \left(\bar{\bf a}_i^{\rm T}{\bf v}_m\right){\bf v}_m\|}-1\nonumber\\
	&= \frac{1}{\sqrt{1 - \sum_{m=1}^{i-1} \left(\bar{\bf a}_i^{\rm T}{\bf v}_m\right)^2}} - 1. 
\end{align}
According to the Taylor's series with Peano form of the remainder, i.e.,
$
f(x) = \frac{1}{\sqrt{1-x}} = 1 + \frac{x}{2} + h(x)x,
$
where $\lim_{x \to 0} h(x) = 0$, \eqref{eq_L1_9} is approximated by
\begin{equation}
	\bar u_{ii} = \Big(\frac12+h(\cdot)\Big) \sum_{m=1}^{i-1} \left(\bar{\bf a}_i^{\rm T}{\bf v}_m\right)^2, \label{eq_L1_10}
\end{equation}
where $h\left(\sum_{m=1}^{i-1} \left(\bar{\bf a}_i^{\rm T}{\bf v}_m\right)^2\right)$ is denoted by $h(\cdot)$ for short.
Following \eqref{eq_L1_3} and using the definition of $\bar{\bf R}$, for $m<i$ we have
\begin{equation}\label{eq_L1_11}
        \bar{\bf a}_i^{\rm T}{\bf v}_m = \bar{\bf a}_i^{\rm T}\bar{\bf a}_m + \sum_{k=1}^m \bar u_{km}\bar{\bf a}_i^{\rm T}\bar{\bf a}_k 
        = \bar r_{mi} + \sum_{k=1}^m \bar u_{km}\bar r_{ki}. 
\end{equation}
Plugging \eqref{eq_L1_11} into \eqref{eq_L1_10}, we have
\begin{align}
	\bar u_{ii} &= \left(\frac12+h(\cdot)\right) \sum_{m=1}^{i-1} \bigg(\bar r_{mi} + \sum_{k=1}^m \bar u_{km}\bar r_{ki}\bigg)^2\nonumber\\
	&=  \left(\frac12+h(\cdot)\right) \bigg(\sum_{m=1}^{i-1}\bar r_{mi}^2 + \sum_{m=1}^{i-1}\bigg(\bigg(\sum_{k=1}^m \bar u_{km}\bar r_{ki}\bigg)^2 
	+ 2\sum_{k=1}^m \bar u_{km}\bar r_{mi}\bar r_{ki}\bigg)\bigg). \label{eq_L1_12}
\end{align}
Because of the symmetry of $\bar{\bf R}$, the first summation in the RHS of \eqref{eq_L1_12} is bounded by $\frac12\|\bar{\bf R}\|_F^2$. Furthermore, the second summation, which is composed of squares and products of $\bar r_{pq}$, must be bounded by $\epsilon_1\|\bar{\bf R}\|_F^2$, where $\epsilon_1$ is a small quantity. 
Consequently, we have
\begin{equation}
	\bar u_{ii} = \bar g_{ii}(\bar{\bf R})\|\bar{\bf R}\|_F^2 \le  \Big(\frac12+h(\cdot)\Big) \Big(\frac12 + \epsilon_1\Big)\|\bar{\bf R}\|_F^2, \label{eq_L1_13}
\end{equation}
where
$
	\lim_{\bar{\bf R}\rightarrow {\bf 0}}\bar g_{ii}(\bar{\bf R}) \le \frac14
$.
Because $h(\cdot)$ tends to $0$ as $\bar{\bf R}$ approaches $\bf 0$. We then complete the first part of the lemma.

Next we will study \eqref{eq_L1_7}. Plugging \eqref{eq_L1_6} and \eqref{eq_L1_11} into \eqref{eq_L1_7}, we have
\begin{align}
	\bar u_{ji} &= -(1+\bar u_{ii})\bigg(\bar r_{ji} + \sum_{k=1}^j \bar u_{kj}\bar r_{ki}
	 +\sum_{m=j}^{i-1}\bigg(\bar r_{mi} + \sum_{l=1}^m \bar u_{lm}\bar r_{li}\bigg)\bar u_{jm}\bigg)\nonumber\\
	&= -(1+\bar u_{ii})\bigg(\bar r_{ji} + \sum_{k=1}^j \bar u_{kj}\bar r_{ki}+\sum_{m=j}^{i-1}\bar u_{jm}\bar r_{mi}
	+\sum_{m=j}^{i-1}\sum_{l=1}^m \bar u_{lm}\bar u_{jm}\bar r_{li}\bigg), \quad\forall j<i.\label{eq_L1_15}
\end{align}
Notice that the summations in \eqref{eq_L1_15}, which are composed of $\bar r_{pq}$, must be bounded by $\epsilon_2\|\bar{\bf R}\|_F$, where $\epsilon_2$ is a small quantity. 
Plugging \eqref{eq_L1_13} into \eqref{eq_L1_15}, we have
\begin{align}\label{eq_L1_16}
	\bar u_{ji} &= -\left(1+\bar g_{ii}(\bar{\bf R}\right)\|\bar{\bf R}\|_F^2)\left(\bar r_{ji} + \epsilon_2\|\bar{\bf R}\|_F\right)\nonumber\\
	&= -\bar r_{ji} + \bar g_{ji}(\bar{\bf R})\|\bar{\bf R}\|_F,
\end{align}
where
$
	\lim_{\bar{\bf R}\rightarrow{\bf 0}}\bar g_{ji}(\bar{\bf R}) = 0.
$
The second part of the lemma is proved.

\bibliographystyle{IEEEtran}
\bibliography{mybibfile}

\clearpage

\end{document}